\documentclass[nofootinbib,aps,superscriptaddress,11pt,twoside]{revtex4}

\usepackage{amsmath,latexsym,amssymb, verbatim, enumerate, graphicx}

\usepackage{theorem}
\newtheorem{definition}{Definition}
\newtheorem{Proposition}[definition]{Proposition}
\newtheorem{Lemma}[definition]{Lemma}

\newtheorem{Theorem}[definition]{Theorem}
\newtheorem{Corollary}[definition]{Corollary}
\newtheorem{conjecture}[definition]{Conjecture}

\def\>{\rangle}
\def\<{\langle}

\def\be{\begin{equation}}
\def\ee{\end{equation}}
\def\bee{\begin{eqnarray*}}
\def\eee{\end{eqnarray*}}

\newcommand{\nc}{\newcommand}
\newcommand{\ket}[1]{|#1\rangle}
\newcommand{\bra}[1]{\langle#1|}
\newcommand{\braket}[1]{|#1\rangle\langle#1|}
\newcommand{\ti}[1]{\tilde{#1}}
\newcommand{\ov}[1]{\overline{#1}}

\newcommand{\initstate}{\psi^{C_1C_2\ldots C_mBR}}
\newcommand{\initstateA}{\psi^{C_1C_2\ldots C_mABR}}

\newcommand{\inputstate}{\psi^{C_1C_2\ldots C_mAB}}
\newcommand{\inputGroupstate}{\psi^{{\cal T}{\ov{\cal T}}ABR}}

\newcommand{\projstate}{\psi_J^{C^1_MB^0_MBR}}

\newcommand{\projstatebraket}{\braket{\psi_{J}}^{C^1_MB^0_MBR}}
\newcommand{\reducstate}{\psi_{J}^{C^1_M R}}

\newcommand{\mergestate}{\psi^{B_MBR}}
\newcommand{\decouplestate}{\tau^{C^1_M} \otimes \psi^{R}}
\newcommand \Tr {\mathrm{Tr}}
\newcommand \calE {{\cal E}}
\newcommand{\sourcestate}{\psi_{C_1C_2\ldots C_mR}}

\newcommand{\cK}{{\cal T}}
\newcommand{\cKbar}{{\overline{\cal T}}}

\newcommand{\cT}{\ov{\cal T}}
\newcommand{\cTo}{\ov{{\cal T}^0}}
\newcommand{\cTu}{\ov{{\cal T}^1}}
\nc{\outputrho}{\rho^{\cK^1A^1_{\cK} \cTu B^1_{\cKbar}A_{\cK}B_{\cKbar} ABRX}}
\newcommand {\X}{{\cal X}}
\newcommand {\Y}{{\cal Y}}
\newcommand {\cX}{{\overline{\cal X}}}
\newcommand {\cY}{{\overline{\cal Y}}}

\nc{\smfrac}[2]{\mbox{$\frac{#1}{#2}$}}
 \nc{\UA}{U^A_{j_{\cK}}}
 \nc{\VB}{V^B_{j_{\cKbar}}}
 \nc{\jK}{j_{\cK}}
 \nc{\jKbar}{j_{\cKbar}}

\def\squareforqed{\hbox{\rlap{$\sqcap$}$\sqcup$}}
\def\qed{\ifmmode\squareforqed\else{\unskip\nobreak\hfil
\penalty50\hskip1em\null\nobreak\hfil\squareforqed
\parfillskip=0pt\finalhyphendemerits=0\endgraf}\fi}

\newenvironment{proof}{\noindent \textbf{{Proof~} }}{\qed\medskip}
\newenvironment{proof+}[1]{\noindent \textbf{{Proof #1~} }}{\qed\medskip}
\begin{document}

\title{One-Shot Multiparty State Merging}

\author{Nicolas Dutil}
\email{ndutil@cs.mcgill.ca}
\affiliation{
    School of Computer Science,
    McGill University,
    Montreal, Quebec, H3A 2A7, Canada
 }

\author{Patrick Hayden}
\email{patrick@cs.mcgill.ca}
\affiliation{
    School of Computer Science,
    McGill University,
    Montreal, Quebec, H3A 2A7, Canada
 }
 \affiliation{
    Perimeter Institute for Theoretical Physics,
    31 Caroline St. N., Waterloo, Ontario,
    N2L 2Y5, Canada
}

%\keywords{entanglement of assistance, quantum state merging, noise}
\begin{abstract} We present a protocol for performing state merging
when multiple parties share a single copy of a mixed state, and analyze the entanglement cost in terms of min- and max-entropies. Our protocol allows for interpolation between corner points of the rate region without the need for time-sharing, a primitive which is not available in the one-shot setting. We also compare our protocol to the more naive strategy of repeatedly applying a single-party merging protocol one party at a time, by performing a detailed analysis of the rates required to merge variants of the embezzling states.
Finally, we analyze a variation of multiparty
merging, which we call \emph{split-transfer}, by considering two
receivers and many additional helpers sharing a mixed state. We give a
protocol for performing a split-transfer and apply it to the
problem of assisted entanglement distillation. \end{abstract} \maketitle
\parskip .75ex

\section{Introduction}
An  important part of quantum information theory is concerned with
the design and analysis of quantum communication protocols.
The subject has flourished
over the past two decades, with early discoveries like teleportation \cite{teleportation} and superdense coding
\cite{superdense} laying the groundwork
for a series of major advances over the last five years. (See, for example,
\cite{SVW,Merge,Hayden001,DW,Yard01,Yard02,zero,hasting}.)
Another early result, Schumacher compression~\cite{Schumacher}, studies the amount of quantum communication required
to transmit to another location a sequence of quantum states
$\ket{\psi^{A}_1}\ket{\psi^{A}_2}\ket{\psi^{A}_3}\ldots$ emitted
by a statistical source. If we assume the
states coming from the source are independent and identically
distributed (i.e an i.i.d source), we get the quantum analogue of
Shannon compression, and the optimal rate of compression is given
by the von Neumann entropy $S(\rho)$ of the density matrix
$\rho=\sum_j p_j \psi_j$ associated with the source
\cite{Schumacher}. This gives an informational meaning to the von
Neumann entropy, whose original definition was motivated by the
desire to extend the Gibbs entropy, a thermodynamical concept, to
the quantum setting. Schumacher compression is often used in more
complex protocols as a preliminary preprocessing step.

Other information theoretic quantities, such as the conditional
von Neumann entropy $S(A|B)_{\psi}$ \cite{Negative} and the
conditional mutual information $I(A : B|R)_{\psi}$ \cite{Yard01},
were only more recently given meaning \cite{SW-nature,Merge,Yard01}. If
we consider an i.i.d. source $S$ emitting an unknown sequence of
states
$\ket{\psi^{AB}_1}\ket{\psi^{AB}_2}\ldots\ket{\psi^{AB}_n}$,
distributed to two spatially separated parties $A$ (Alice) and $B$
(Bob), an interpretation of $S(A|B)_{\psi}$ can be obtained
\cite{SW-nature,Merge} as the optimal rate at which pure
entanglement needs to be consumed in order to transfer the entire
sequence to Bob's location. Whenever $S(A|B)_{\psi}$ is negative,
it is understood that entanglement is gained instead of consumed
and that the transfer can be accomplished using only local operations and
classical communication (LOCC). The first protocol
\cite{Merge,SW-nature} for achieving this task, also known as
state merging, was based on a random measurement strategy, a popular approach when designing quantum communication protocols. Examples of other tasks that can be achieved using this approach are distributed compression \cite{Hayden002} and assisted distillation \cite{dfm}. In assisted distillation, $m$ helpers $C_1,
C_2,\ldots, C_m$ and two recipients $A$ and $B$ share a
multipartite pure state $\psi^{C_1C_2\ldots C_mAB}$, and the objective
is to extract an optimal amount of pure entanglement between $A$
and $B$ by using LOCC operations and classical information
broadcasted by the $m$ helpers $C_1, C_2,\ldots, C_m$.

Both distributed compression and assisted distillation involve
multiple parties (i.e more than two) sharing a multipartite state
$\psi$. In the case of distributed compression, we can use the state merging
primitive to perform compression at an optimal rate by transferring each sender's share one at a time \cite{Merge}. This strategy
will work for any rate which is a corner point of the boundary of the rate region associated with distributed compression. To achieve compression at rates which are not corner points, however, we need to use a time-sharing strategy as the decoding operation performed by the receiver can only recover the shares one at a time. One contribution of this paper is to present a protocol for the more general task of multiparty state merging which will eliminate the need for time-sharing. That is, we consider $m$ senders $C_1, C_2,\ldots, C_m$ and a decoder $B$ sharing a multipartite mixed state $\psi^{C_1C_2\ldots C_mB}$,
potentially with additional entanglement in the form of EPR pairs
(ebits) distributed between the decoder and each of the $m$
senders. Given many copies of the input state $\psi^{C_1C_2\ldots C_mB}$, the task is to transfer the shares $C_1, C_2,\ldots,
C_m$ to the receiver $B$ with high fidelity using only LOCC
operations.

If only a single copy of the state $\psi^{C_1C_2\ldots C_mB}$ is
available to the parties, we can use our protocol to achieve merging within an error tolerance $\epsilon$ if we distribute enough initial entanglement between each of the senders and the receiver. In this regime, a more naive strategy consisting of repeatedly applying a one-shot state merging protocol \cite{Berta} on one sender at a time will generally require more initial entanglement to perform the state transfer than does our protocol. In addition, this strategy only yields a handful of achievable combinations of entanglement costs and does not permit interpolating between them. A full characterization of the entanglement cost in the one-shot regime when $m=1$
was performed by \cite{Berta}  using smooth min- and max-entropies. By applying the random
measurement strategy of \cite{Merge} and by using the min- and max-entropy formalism of \cite{Renner02}, we generalize some of the results of \cite{Berta} to the multipartite case ($m \geq 2$). This work complements other recent attempts to study quantum information theory in the one-shot setting~\cite{Renner01,Berta02,Datta1,Datta,Renes}.

To perform assisted distillation in the context of multiple
parties, we introduce a second decoder, whom we label $A$ (Alice),
and consider a variation of multiparty merging. Given a partition
of the helpers $C_1, C_2\ldots, C_m$ into a set ${\cal T}$ and its
complement $\overline{{\cal T}}:=\{C_1C_2\ldots
C_m\}\backslash{\cal T}$, we want to transfer the shares ${\cal
T}$ and ${\overline{\cal T}}$ to the locations of the decoders $A$
and $B$ respectively. We call this task a \emph{split-transfer} of
the state $\psi^{C_1C_2\ldots C_mB}$. A protocol for performing a
split-transfer can be obtained by using the random measurement
strategy on $C_1, C_2,\ldots, C_m$, followed by appropriate
decodings $U_A$ and $V_B$ by the decoders $A$ and $B$. The optimal
achievable rate for assisted distillation was found in
\cite{Merge} by using a recursive argument. By using a split-transfer protocol, we give a simpler demonstration which does not rely on a recursive argument.

\textbf{Structure of the paper:} In Section~\ref{sec:defns}, we introduce the
definition for multiparty merging of a state
$\psi^{C_1C_2\ldots C_mBR}$ and review the known results for the i.i.d setting. In Section~\ref{sec:merging-many}, we formulate a condition that a set of instruments performed by
the senders $C_1,C_2,\ldots,C_m$ must satisfy in order to accomplish
merging within a fixed error tolerance. In Section~\ref{sec:random-meas}, we consider
random measurements performed by the senders $C_1, C_2,\ldots,
C_m$ and prove an upper bound to the quantum error when a single
copy of the input state is available. We analyze the asymptotic setting
in Section~\ref{sec:iid}, recovering the main theorem of Section~\ref{sec:defns} without the need for time-sharing. In Section~\ref{sec:one-shot}, we reformulate the bounds obtained in Section~\ref{sec:merging-many} in terms of
min-entropies and give necessary and sufficient conditions for
merging in the one-shot regime.  Section\ref{sec:example} is devoted to analyzing the rates achievable for variants of the embezzling states, comparing our protocol to a strategy of merging the shares one at a time. We introduce a \emph{split-transfer} of the state $\psi^{C_1C_2\ldots
C_mB}$ in Section~\ref{sec:split-transfer} and show the existence of a protocol for
performing this task. We use this protocol to recover the optimal
distillation rate for the problem of assisted distillation.
Appendices, containing relevant folklore material, appear at the end.

\textbf{Notation:} In this paper, we restrict our attention to
finite dimensional Hilbert spaces. Quantum systems under
consideration will be denoted $A, B,\ldots,$ and are freely
associated with their Hilbert spaces, whose (finite) dimensions
are denoted $d_A, d_B$, etc... If $A$ and $B$ are two Hilbert
spaces, we write $AB \equiv A \otimes B$ for their tensor product
and write $A^n$ for the tensor product $\bigotimes_{i=1}^n A$. If we have $m$ Hilbert spaces $A_1, A_2, \ldots, A_m$, we write $A_M$ for the tensor product $\bigotimes^m_{i=1}A_i$. An ancilla augmenting the system $A_i$ is denoted as $A^0_i$. We write $A^0_M$ for the tensor
product $\bigotimes^m_{i=1} A^0_i$ of $m$ ancillas $A^0_1, A^0_2,\ldots, A^0_m$. The maximally entangled
state $\frac{1}{\sqrt{K_i}}\sum_{k=1}^{K_i} \ket{k}\ket{k}$ of Schmidt
rank $K_i$ is denoted as $\ket{\Phi^{K_i}}$. We
write $\Phi^K$ for the density operator $\braket{\Phi^{K_1}}\otimes
\braket{\Phi^{K_2}} \otimes \ldots \otimes \braket{\Phi^{K_m}}$. The space of linear
operators acting on the Hilbert space $A$ is denoted by ${\cal L}(A)$. The
identity operator acting on $A$ is denoted by $I^A$. The symbol
$\mathrm{id}_A$ denotes the identity map acting on ${\cal L}(A)$. Unless otherwise
stated, a "state" can be either pure or mixed. The symbol for such
a state (such as $\psi$ and $\rho$) also denotes its density
operator. The density operator $\braket{\psi}$ of a pure state
will frequently be written as $\psi$.  We denote by $\tau_A$ the
maximally mixed state of dimension $d_A$. We write $S(A)_{\psi} =
-\Tr(\psi^A \log \psi^A)$ to denote the von Neumann entropy of a
density matrix $\psi^A$ for the system $A$. The function
$F(\rho,\sigma) := \Tr \sqrt{\rho^{1/2}\sigma\rho^{1/2}}$ is the
fidelity \cite{uhlmann:fid} between the two states $\rho$
and $\sigma$. The trace norm of an operator, $\|X\|_1$ is defined
to be $\Tr|X| = \Tr \sqrt{X^{\dag}X}$. We will use the terms
"receiver" and "decoder" interchangeably throughout the following
sections.

\section{Definitions and Main Result} \label{sec:defns}

For a bipartite state $\rho^{AB}$, the operation known as
\emph{quantum state merging} can be viewed in two different ways.
The original formulation of the problem \cite{Merge} was in terms
of a statistical source emitting (unknown) states
$\ket{\psi^{AB}_1},\ket{\psi^{AB}_2},\ldots,$ with average density
operator $\rho^{AB}$ assumed to be known by the parties. The
objective was then to transfer the entire sequence to the location
of the decoder (Bob) using as little quantum communication as
possible. An equivalent view of the problem is to consider a
purification $\psi^{ABR}$ of the density matrix $\rho^{AB}$, and
regard the process of merging as that of transferring all the
correlations between Alice's share and the purification system $R$
to the location of the decoder $B$. This means decoupling Alice's
system from the reference $R$, while leaving the state $\psi^R$
intact (up to some arbitrarily small perturbation) in the process.
The receiver will hold a purification $\phi^{BR}$ of the system
$R$, and since all purifications are equivalent up to an isometry
on the purification system, he can recover the original state
$\psi^{ABR}$ by applying an appropriate isometry to $\phi^{BR}$.
Additional entanglement between Alice and Bob might be distilled
in the process.

To analyze the multipartite scenario, where $m$
senders and a decoder/receiver share a state $\initstate$, with
purifying system $R$, we adopt the second view and look at the
transformations that can be performed by the senders on their shares to allow the
receiver to recover the purified state $\initstate$ with high
fidelity. The resources at their disposal will be pure entanglement,
in the form of maximally entangled states shared between each of the senders $C_1, C_2, \ldots, C_m$ and the receiver Bob, and
noiseless classical channels, which will be used to
transmit measurement outcomes to the receiver. Any
transformation applied by the senders will need to decouple the
reference $R$ from the senders' shares $C_1,C_2,\ldots,C_m$ and leave the
reference unchanged. Otherwise, the receiver might hold a
purification that cannot be taken, by means of an isometry, to the
original state. If each of the senders $C_1, C_2, \ldots, C_m$ perform an incomplete measurement, described by Kraus operators $P_i$ mapping $C_i$ to a subspace $C^1_i$, we would want each outcome state $\psi^{C^1_1C^1_2\ldots
C^1_mR}_{J}$ to have a product form
 \begin{equation}
    \psi^{C^1_1C^1_2\ldots C^1_m R}_J \approx \psi^{C^1_1C^1_2\ldots C^1_m}_J \otimes \psi^{R},
 \end{equation}
where $J= (j_1, j_2, \ldots, j_m)$ are the measurement outcomes.
The states $\{\psi^{C^1_1C^1_2\ldots C^1_m}_J\}$ could be entangled
and/or contain classical correlations between some of the subsystems $C^1_1, C^1_2,\ldots,C^1_m$. In this paper, we will primarily be
concerned with extracting pure bipartite entanglement, in the form of
maximally entangled states $\frac{1}{\sqrt{K}}\sum^K_{i=1}\ket{k}\ket{k}$ shared between the senders and the decoder, and thus, we will
further impose that the operations applied by the senders
destroy all correlations existing with the other senders' shares. That is,
we want
 \begin{equation}
    \psi^{C^1_1C^1_2\ldots C^1_m R}_J \approx \tau^{C^1_1} \otimes \tau^{C^1_2} \otimes \ldots \otimes \tau^{C^1_m} \otimes \psi^{R},
 \end{equation}
where $\tau^{C^1_i}$ is the maximally mixed state of dimension
$L_i$ on the subspace $C^1_i$.  With this assumption in
mind, we can give a definition of a multiparty state merging for a
state $\initstate$.

\begin{figure}[t]
  \centering
    \includegraphics[width=0.9\textwidth]{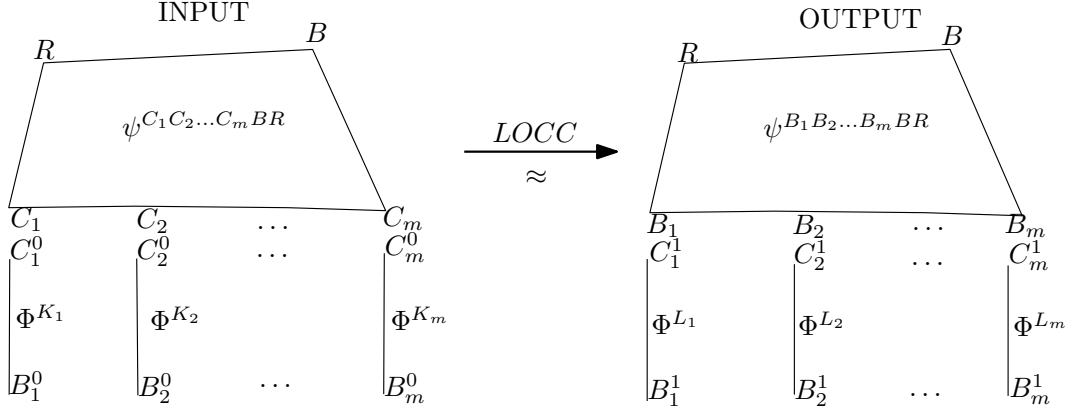}
\caption{Picture of the initial and final steps of a multiparty
state merging protocol.
%Initially, the state is shared between $m$
%parties $C_1,C_2,\ldots,C_m$ and a receiver (Bob), with purifying
%system $R$. Additionally, each party $C_i$ shares a maximally
%entangled state $\Phi^{K_i}$ of Schmidt rank $K_i$ with the
%receiver. After the protocol is applied, we expect the output
%state of the protocol to be close to the state depicted on the
%right hand side of the picture. That is, all correlations present
%between the parties and the reference $R$ are transferred to the
%receiver, with Bob holding a purification $\psi^{B_1B_2\ldots
%B_mR}$ of the state $\psi^{R}$ which corresponds to the original
%state $\psi^{C_1C_2\ldots C_mBR}$, with the systems
%$B_M:=B_1B_2\ldots B_m$ substituted for $C_M:=C_1C_2\ldots C_m$.
%Additional maximally entangled states $\Phi^{L}:=\Phi^{L_1}\otimes
%\Phi^{L_2} \otimes \ldots \Phi^{L_m}$ are also
%generated.
}\label{fig:statemerging}
\end{figure}

Let $\Lambda^m_{\rightarrow} : C_MC^0_M \otimes BB^0_M \rightarrow
C^1_M \otimes B^1_MBB_M$ be an LOCC quantum operation performed by
the senders $C_1,C_2,\ldots,C_m$ and the decoder $B$. Initially, each sender
$C_i$ is given an ancilla $C^0_i$. The receiver also has
ancillas $B^0_M:=B^0_1B^0_2\ldots B^0_m$, with $d_{B^0_i} =
d_{C^0_i}$, and $B^1_M:=B^1_1B^1_2\ldots B^1_m$ with $d_{B^1_i} = d_{C^1_i}$. Before the map $\Lambda^m_{\rightarrow}$ is applied, the systems $C^0_M$ and $B^0_M$ will hold maximally entangled states $\Phi^{K_1} \otimes \Phi^{K_2} \otimes \ldots \otimes \Phi^{K_m}$, where $\Phi^{K_i}$ has Schmidt rank $K_i = d_{C^0_i}$ and is shared between the sender $C_i$ and the receiver $B$.
After the map $\Lambda^m_{\rightarrow}$ is applied, the senders
share a subsystem $C^1_M:=C^1_1C^1_2\ldots C^1_m$ of $C_M$, and the receiver holds three systems:
$B$, $B^1_M$ and $B_M$, with $B_M$ being an ancillary system of dimension of the same size as
the system $C_M$.

This operation will implement merging, as illustrated in Figure
\ref{fig:statemerging}, if the output state of the map
$\mathrm{id}_R \otimes \Lambda^m_{\rightarrow}$ is approximately a
tensor product of the initial state $\initstate$ and maximally
entangled states $\Phi^L:=\Phi^{L_1} \otimes \Phi^{L_2} \otimes
\ldots \otimes \Phi^{L_m}$ shared between the senders and the
decoder. Each $\Phi^{L_i}$ is a maximally entangled state of
Schmidt rank $L_i$ on the tensor space $C^1_iB^1_i$. More
formally, we have

\begin{definition}[$m$-Party State Merging]
Let $\Lambda^m_{\rightarrow}$ be defined as in the previous
paragraphs. We say that $\Lambda^m_{\rightarrow}$ is an
\emph{$m$-party state merging protocol} for the state $\initstate$ with error $\epsilon$ and
entanglement cost $\overrightarrow{E}:= (\log K_1 -\log L_1,\log K_2 -\log L_2, \ldots, \log K_m - \log L_m)$ if
\begin{equation}
 \biggl \| (\mathrm{id}_{R} \otimes \Lambda^m_{\rightarrow})(\initstate \otimes \Phi^K) -
\mergestate \otimes \Phi^L \biggr \|_1 \leq
\epsilon,
\end{equation} where the state $\mergestate$ corresponds to the initial state
$\initstate$ with the system $B_M$ substituted for $C_M$. If we
are given $n$ copies of the same state, $\psi = (\sigma)^{\otimes
n}$, the entanglement rate $\overrightarrow{R}(\sigma)$ is defined
as $\overrightarrow{R}(\sigma) :=
\frac{1}{n}\overrightarrow{E}(\psi)$.
\end{definition}

%We define a a one-way forward LOCC operation ${\cal M}:\ti{C}_1
%\ti{C}_2...\ti{C}_m \otimes \ti{B} \rightarrow C^1_M \otimes
%B^1_M\ti{B}\ti{B}_M $ as:n $m$-party state merging of $\psi$ with
%error $\epsilon$, if it is LOCC, and  \be \bigg \|({\cal M}
%\otimes I_R)(\initstate \otimes \phi_K) - \mergestate \otimes
%(\phi_L)_{C^1_M B^1_M} \bigg \|_1 \leq \epsilon, \ee where $\phi_K
%= \phi_{K_1} \otimes \phi_{K_2} \otimes \ldots \otimes \phi_{K_m}$
%and $\phi_L = \phi_{L_1} \otimes \phi_{L_2} \otimes \ldots \otimes
%\phi_{L_m}$ are such that each $\phi_{K_i}$ (resp. $\phi_{L_i}$)
%is a maximally entangled state of dimensions $K_i$ (resp. $L_i$)
%on the systems $C^0_i B^0_i$ (resp. $C^1_i B^1_i)$. We use the
%shorthand notation $B^0_M$ for the system $B^0_1 B^0_2...B^0_m$.
%The systems $C^1_M$ and $B^1_M$ are similarly defined. Lastly, the
%system $\ti{B}_M$ is a local ancilla of Bob's of the same size as
%$\ti{C}_1\ti{C}_2...\ti{C}_m \equiv \ti{C}_M$.

%The number $\log{K_i}-\log{L_i}$ corresponds to the entanglement cost of party $\ti{C}_i$ for the protocol. The
%total entanglement cost of the protocol is given by $\sum_i (\log{K_i} - \log{L_i})$. In the case of many copies
%of the same state, $\Psi = \psi^{\otimes n}$, we call $\frac{1}{n}( \log{K_i}-\log{L_i}) \equiv R_i$ the
%entanglement rate of the protocol for the party $\ti{C}_i$. An $m$-tuple $R \equiv (R_1,R_2,...,R_m)$ is an
%achievable rate if for all $\epsilon > 0$ and sufficiently large values of $n$, there exist $m$-party merging
%protocols of rates approaching $R$ and error at most $\epsilon$.
%\end{definition}
Before stating the main theorem, we need to define what it means
for a rate-tuple $\overrightarrow{R}$ to be achievable for
multiparty merging using LOCC operations.

\begin{definition}[The Rate Region]\label{def:rateregion}
We say that the rate-tuple
$\overrightarrow{R}:=(R_1,R_2,\ldots,R_m)$ is achievable for
multiparty merging of the state $\psi^{C_1C_2\ldots C_m B}$ if,
for all $\epsilon
> 0$, we can find an $N(\epsilon)$ such that for every $n \geq
N(\epsilon)$ there exists an $m$-party state merging
protocol $\Lambda^m_{\rightarrow}$ acting on $\psi^{\otimes n} \otimes \Phi^{K^n}$, with error
$\epsilon$ and entanglement rate $\overrightarrow{R_n}:= \frac{1}{n}(\log K^n_{1} -\log L^n_1, \log K^n_2 -\log L^n_2,\ldots, \log K^n_m - \log L^n_m)$ approaching
$\overrightarrow{R}$. We call the closure of the set of achievable
rate-tuples the rate region.
\end{definition}

Suppose $m$ systems $C_1, C_2,\ldots, C_m$ in the state $\psi^{C_1C_2\ldots C_mR}$, where $R$ is a purifying system, are distributed to $m$ senders, spatially separated from each other. To recover the purified state $\psi^{C_1C_2\ldots C_mR}$ at the receiver's end, a task called distributed compression, the senders need an enough supply of initial entanglement. It was shown in \cite{Merge} that the rate region associated with distributed compression is characterized by the inequalities
\begin{equation}
\sum_{i \in {\cal T}} R_i \geq S({\cal T}|\overline{{\cal T}})_{\psi} \phantom{==} \quad \text{for all nonempty sets } {\cal T}\subseteq \{1,2,\ldots,m\}.
\end{equation}
Here, the symbol ${\cal T}$ also denotes the tensor product space ${\cal T}:=\bigotimes_{i \in {\cal
T}} C_i$ associated with the set ${\cal T}$. The set $\overline{{\cal T}}$ is defined as $\{1,2,\ldots,m\}\backslash {\cal T}$ and the tensor product space $\overline{{\cal T}}$ as $\bigotimes_{i \in {\overline{\cal
T}}} C_i$. If a rate-tuple $(R_1, R_2, \ldots, R_m)$ is achievable for distributed compression and some of the rates $R_i$ are negative, then the senders $C_i$ will be able to transfer their shares to the receiver using only LOCC operations, and furthermore, they will gain a potential for future communication in the form of maximally entangled states. Allowing the receiver to have side information $B$ as well, leads to a similar set of equations describing the rate region associated with the task of multiparty state merging.

\begin{Theorem}[$m$-Party Quantum State Merging \cite{Merge}] \label{thm:statemerging}
Let $\initstate$ be a pure state shared between $m$ senders
$C_1,C_2,\ldots,C_m$ and a receiver Bob, with purifying system
$R$. Then, the rate
$\overrightarrow{R}:=(R_1,R_2,\ldots,R_m)$ is achievable for
multiparty merging  iff the inequality \be \label{eq:Mainthm3} \sum_{i \in {\cal
T}} R_i \geq S({\cal T}|\overline{{\cal T}}B)_{\psi} \ee holds for
all non empty subsets ${\cal T} \subseteq \{1,2,...,m\}$.
\end{Theorem}

The theorem was proved in \cite{Merge} by showing that the corner points of the region are achievable and then using time-sharing to interpolate between them. In addition to recovering the result without time-sharing, we will extend it to the one-shot setting. Time-sharing, which consists of partitioning a large supply of states and applying different protocols to each subset, is impossible if only a single copy of a state is available. Instead, we will construct a single protocol which merges all shares at once.

\section{Conditions for Merging Many Parties} \label{sec:merging-many}

%Given a pure state $\initstate$, consider a protocol where each party $C_i$ performs an incomplete measurement
%given by Kraus operators $P^j_i$ mapping $\ti{C_i}$ to $C^1_i$ (in our actual solution, it will be a von Neumann
%measurement followed by a unitary). Given that the outcomes were $j_M \equiv j_1j_2...j_m$, the initial state
%$\psi$ collapses to a state which we will denote by $\projstate$, \be |\projstate\rangle =
%\frac{1}{\sqrt{p_{j_M}}}(P^{j_1}_1 \otimes P^{j_2}_2 \otimes \ldots \otimes P^{j_m}_m \otimes
%I_{\ti{B}\ti{R}})\ket{\initstate}, \ee where $p_{j_M}$ is the probability of obtaining outcome $j_M$, \be
%p_{j_M} = \bra{\psi}((P^{j_1}_1)^{\dag}P^{j_1}_1 \otimes (P^{j_2}_2)^{\dag}P^{j_2}_2 \otimes \ldots \otimes
%(P^{j_m}_m)^{\dag}P^{j_m}_m \otimes I_{\ti{B}\ti{R}}) \ket{\psi} \\
%\ee
Let's try to construct an LOCC operation $\Lambda^m_{\rightarrow} : C_M C^0_M \otimes BB^0_M \rightarrow C^1_M \otimes B^1_MBB_M$ that when applied to the state $\psi^{C_1C_2\ldots C_mB} \otimes \Phi^{K}$ will achieve merging, and destroy all existing correlations between the senders' shares $C_1,C_2,\ldots,C_m$ at the same time. It can be seen as a three step process:
\begin{itemize}
 \item First, each sender $C_{i}$ applies a quantum instrument ${{\cal I}_i} :=\{\calE^i_j\}^X_{j=1}$ to his share of the state $\initstate \otimes \Phi^K$. This will yield both quantum and classical outputs. Each operator $\calE^i_j$ for the instrument ${\cal I}_i$ is completely positive, and maps the space $C_iC^0_i$ to the subspace $C^1_i$.
 \item Secondly, the senders $C_1,C_2,\ldots,C_m$ send their classical outputs $J:=(j_1,j_2,\ldots,j_m)$ to the decoder $B$.
 \item Finally, the decoder will use his side information $\psi^{B}$, his share of the maximally entangled states $\{\Phi^{K_i}\}$, and the classical information $J$ to perform a decoding operation ${\cal D}_{J}:BB^0_M\rightarrow B^1_MBB_M$ (i.e a trace-preserving completely positive map (TP-CPM)) and recover the state $\initstate \otimes \Phi^L$.
\end{itemize}
%The state of the systems $BB_MC^1_MB^1_M$ after the operation $\Lambda^m_{\rightarrow}$ is performed can be described as:
%\begin{equation}
%  \varphi^{BB_MC^1_MB^1_MR{\cal X}} := \sum_{J \in
%  {\cal X}} [(\mathrm{id}^R \otimes \calE^{J} \otimes {\cal
%  D}^{J})(\initstate \otimes \Phi^K)]^{BB_MB^1_MC^1_M} \otimes \braket{J}^{\cal X},
%\end{equation}where $\calE^J := {\calE}_{j_1}
%\otimes {\calE}_{j_2} \otimes \ldots \otimes {\calE}_{j_m}$ and ${\cal X}$ is an ancillary system held by the receiver which %contains the value of the measurement outcome. Our map $\Lambda^m_{\rightarrow}$ will accomplish merging if the state %$\varphi^{BB_MC^1_MB^1_MR}$ is $\epsilon$-close in distance to a tensor product of the original state $\initstate$ and %additional entanglement $\Phi^L$ between the systems $B^1_M$ and $C^1_M$.
The state of the systems $C^1_MB^0_MBRX$ after steps 1 and 2 are performed can be written as:
 \begin{equation}\label{eq:LOCC}
   \begin{split}
     \psi^{C^1_MB^0_MBRX} &:= \sum_{J:=j_1j_2\ldots j_m} [(\mathrm{id}^{B^0_MBR} \otimes \calE_{J})(\initstate \otimes \Phi^K)]^{C^1_MB^0_MBR} \otimes \braket{J}^X \\
   &=\sum_J p_J \psi_J^{C^1_MB^0_MBR} \otimes \braket{J}^X, \\
   \end{split}
 \end{equation} where $\calE_J := {\calE}^1_{j_1} \otimes {\calE}^2_{j_2} \otimes \ldots \otimes {\calE}^m_{j_m}$ and $\psi_J^{C^1_MB^0_MBR}$ is the normalized state given by $(\mathrm{id}^{B^0_MBR} \otimes \calE_{J})(\initstate \otimes \Phi^K)$. The system $X$ is an ancillary system held by the receiver which contains the classical outputs of the instruments ${\cal I}_1,{\cal I}_2,\ldots,{\cal I}_m$. If we restrict the operators $\calE^i_j$ to consist of only one Kraus operator (i.e $\calE^i_j(\rho)=A^i_j \rho (A^i_j)^{\dag}$ for all $i,j$) and to satisfy $\sum_j (A^i_j)^{\dag}A^i_j = I^{C_i}$, the outcome states $\{\psi^{C^1_MB^0_MBR}_J\}$ are pure and are the result of performing $m$ incomplete measurements, one for each sender $C_i$.

After the senders have finished performing their instruments, we
would ideally like for the state $\psi^{C^1_MR}_J$ to be in the
product form \begin{equation}
 \begin{split}
\reducstate &=
\tau^{C^1_1} \otimes \tau^{C^1_2} \otimes \ldots \otimes
\tau^{C^1_m} \otimes \psi^{R}, \\
&= \tau^{C^1_M} \otimes \psi^R, \\
 \end{split}
\end{equation}
where $\tau_{C^1_i}$ is the maximally mixed state of dimension
$L_i=d_{C^1_i}$ on the space $C^1_i$. Suppose, for the moment, that this
property is satisfied for all $\{\reducstate\}$. Then, the state
$\projstate$ purifies $\tau^{C^1_M} \otimes \psi^R$, with
purification systems $B^0_MB$. Another purification of
$\tau^{C^1_M} \otimes \psi^R$ is also given by $\Phi^{L} \otimes
\mergestate$, where the state $\mergestate$ corresponds to the
original state $\initstate$ with the system $B_M$ substituted for
$C_M$. It follows from the Schmidt decomposition that these two
purifications are related by an isometry $U_J: B^0_MB \rightarrow
B^1_MBB_M$ on Bob's side such that
\begin{equation}
 \begin{split}
(I^{C^1_MR} \otimes U_J) \projstate (I^{C^1_MR} \otimes U_J)^{\dag} &= \Phi^{L_1}
\otimes \Phi^{L_2} \otimes \ldots \otimes \Phi^{L_m}
\otimes \mergestate,  \\
 &= \Phi^{L} \otimes \mergestate. \\
 \end{split}
\end{equation}
Hence, if the senders can perfectly decouple their systems $C_MC^0_M$ from
the reference, their "$R$-entanglement" will be transferred to
Bob's location. Furthermore, by applying $U_J$, the receiver will
recover the original state and distill some pure bipartite entanglement.

The previous scenario was ideal, and in general, will not be feasible for most states $\initstate$. Hence, we relax our decoupling requirement and accept that the measurements performed by the senders will perturb the reference $\psi^{R}$ up to some tolerable amount, and that a small dose of correlations between the senders' shares might still be present. In more formal terms, we have

\begin{Proposition}[Compare to Proposition 4 of \cite{Merge}] \label{prop:mergeCond}
 Let $\psi_J^{C^1_MB^0_MBR}$ be defined as in eq.~(\ref{eq:LOCC}), with reduced density matrix
 $\psi_J^{C^1_MR}$. Define the following quantity:
  \begin{equation}\label{quanterror:eq}
    Q_{{\cal I}}(\initstate \otimes \Phi^K) := \sum_{J} p_{J} \biggl \| \psi_J^{C^1_MR} -
    \tau^{C^1_M} \otimes \psi^{R}\biggr \|_1,
  \end{equation}where $p_{J}$ is the probability of obtaining the state $\psi_J^{C^1_MB^0_MBR}$ after all the senders have performed their instruments.
If $Q_{\cal I}(\initstate \otimes \Phi^K) \leq \epsilon$, then there exists an LOCC operation $\Lambda^m_{\rightarrow}$ which is an $m$-party state merging protocol for the state $\initstate$ with error $2\sqrt{\epsilon}$ and entanglement cost $\overrightarrow{E} = (\log K_1-\log L_1, \log K_2 -\log L_2, \ldots, \log K_m - \log L_m)$.
\end{Proposition}

\begin{proof} The proof of the above statement is very similar to the proof of Proposition 4 in \cite{Merge}. We give the full proof here for completeness. Using Lemma \ref{Lemma:relation} (see Appendix~\ref{app:facts}), we have
\begin{equation}
\sum_{J} p_{J} F(\reducstate,\decouplestate) \geq 1 - \frac{\epsilon}{2}.
\end{equation}
By Ulhmann's theorem, we know there exist an isometry (i.e. a
decoding) $U_J:B^0_MB \rightarrow B^1_MBB_M$ implementable by Bob
such that
\begin{equation}
F(\reducstate,\decouplestate) = F\biggl ((I^{C^1_MR} \otimes
U_J)\projstate (I^{C^1_MR} \otimes U_J)^{\dag},\Phi^{L} \otimes \mergestate \biggr ).
\end{equation}
Thus, using the concavity of $F$ (see \cite{Nielsen} for a proof) in its first argument, we have
\begin{equation}
  \begin{split}
  F(&\psi^{C^1_MB^1_MB_MBR}, \Phi^{L} \otimes \mergestate)\\
  %& \phantom{======} \geq \sqrt{\sum_{J} p_{J} F^2\biggl ((I^{C^1_MR} \otimes U_J)\projstate (I^{C^1_MR} \otimes U_J)^{\dag},\Phi^{L} \otimes \mergestate \biggr )}\\
  & \phantom{======} \geq  \sum_{J}p_{J} F\biggl ((I^{C^1_MR} \otimes
U_J)\projstate (I^{C^1_MR} \otimes U_J)^{\dag},\Phi^{L} \otimes \mergestate \biggr )  \\
  & \phantom{======} \geq 1- \frac{\epsilon}{2}, \\
   \end{split}
\end{equation} where
  \begin{equation}
    \begin{split}
   \psi_{C^1_MB^1_MBB_MR} &= (\mathrm{id}_{R} \otimes \Lambda^m_{\rightarrow})(\initstate \otimes \Phi^K) \\
    &:= \sum_{J}p_{J} (I^{C^1_MR} \otimes U_J)\projstatebraket (I^{C^1_MR} \otimes U_J)^{\dag} \\
    \end{split}
  \end{equation}
is the output state of the protocol. Using the relation between fidelity and trace distance
once more, we arrive at
 \begin{equation}
 \biggl \| \psi^{C^1_M B^1_MBB_MR} -  \Phi^L \otimes \mergestate \biggr \|_1 \leq 2\sqrt{\epsilon-\epsilon^2/4} \leq 2\sqrt{\epsilon}.
 \end{equation}
\end{proof}

\section{One-Shot Merging by Random Measurements} \label{sec:random-meas}

One possible strategy for decoupling the system $C_i$ from the
reference $R$ and the other systems $\{C_j : j \neq i\}$ is to
perform a random von Neumann measurement on $C_i$ with
$N_i=\lfloor \frac{d_{C_i}}{L_i}\rfloor$ projectors of rank $L_i$,
and a little remainder, followed by a unitary $U_i$ mapping the
outcome state to a subspace $C^1_i$ of $C_iC^0_i$. For such
measurements, we can bound the quantum error $Q_{\cal
I}(\initstate \otimes \Phi^K)$ as follows:
%Now that we know that a sufficient condition for accomplishing
%multiparty merging is for the instruments ${\cal I}_1,{\cal I}_2$
%to decouple the system $C_i$ from , while leaving the reference
%intact, we must prove that such instruments does in fact exist. In
%this next result, we show that if each party

\begin{Proposition}[One-Shot Multiparty Merging] \label{prop:isometry}
Let $\initstate$ be a multipartite state, with local dimensions
$d_{B},d_{R}$ and $d_{C_i}, 1 \leq i \leq m$, and let $\Phi^K$ be
some additional pure entanglement, as defined in the previous
sections, shared between the receiver and the senders. For each
sender $C_i$, there exists an instrument ${\cal I}_i$ consisting of
$N_i := \lfloor \frac{d_{C_i}K_i}{L_i} \rfloor$ CP maps \be
{\calE^i_j}(\rho):=P^i_j \rho (P^i_j)^{\dag} \phantom{==} 1\leq j
\leq N_i, \ee where $P^i_j:C_i \rightarrow C^1_i$ is a partial
isometry of rank $L_i$ (i.e. $(P^i_j)^{\dag} P^i_j $ is a
projector onto an $L_i$ dimensional subspace of $C_i$), and one
map $\calE^i_0(\rho):=P^i_0 \rho (P^i_0)^{\dag}$, where $P^i_0$ is
of rank $L_i'= d_{C_i}K_i - N_i L_i < L_i$, such that the overall
quantum error $Q_{{\cal I}}(\initstate \otimes \Phi^K)$ is bounded
by
\begin{equation} \label{eq:upperbound}
  \begin{split}
  Q_{\cal I}(\initstate \otimes \Phi^K) &\leq 2 \sum_{\substack{{\cal T} \subseteq \{1,2,...,m\} \\ {\cal T} \notin \emptyset}} \prod_{i \in {\cal T}}\frac{L_i}{d_{C_i}K_i} + 2\sqrt{d_{R} \sum_{\substack{{\cal T} \subseteq \{1,2,\ldots,m\} \\ {\cal T}
\neq \emptyset}}  \prod_{i \in {\cal T}} \frac{L_i}{K_i} \Tr \bigg [ (\psi^{R{\cal T}})^2 \bigg ]} =: \Delta_{\cal I}, \\
  \end{split}
\end{equation}
and there is a merging protocol with error at most $2\sqrt{\Delta_{\cal I}}$. In fact, for each sender $C_i$, if we
perform a random von Neumann measurement on $C_i$ followed by a
unitary $U$ mapping the outcome state to a subspace $C^1_i$, the left hand side of eq.~(\ref{eq:upperbound}) is
bounded from above on average by the right hand side.
\end{Proposition}
\newcommand \omeg {\omega^{C^1_MR}}
\newcommand \omegapos {\omega^{C'^1_MR'}}
\newcommand \omegsquare {(\omega^{C^1_MR})^2}
\newcommand \decouple {\frac{L}{d_{C_M}} \tau^{C^1_M} \otimes \psi^{R}}
\newcommand \decoupled {\frac{L}{d_{C_M}K} \tau^{C^1_M} \otimes \psi^{R}}

\newcommand {\expec}[1]{E \bigg [ #1 \bigg ]}
\newcommand {\smexpec}[1]{E[#1]}

To prove this proposition, we will need the following technical
lemma, which generalizes Lemma 6 in \cite{Merge} to the case of
$m$ senders. The proof will follow a similar line of reasoning.
\begin{Lemma}[Compare to Lemma 6 in \cite{Merge}] \label{Lemma:rdmisometry}
For each sender $C_i$, let $P_i: C_i \rightarrow C^1_i$ be a random partial isometry of rank $L_i$.
One way to construct such an isometry is to fix some rank $L_i$-projector $Q_i$ onto a subspace $C^1_i$ of $C_i$ and precede it with a Haar
distributed unitary $U_i$ on $C_i$ (i.e $P_i := Q_i U_i$). Define the subnormalized density matrix
\begin{equation}
 \omeg(U)  := (Q_1U_1 \otimes Q_2U_2 \otimes \ldots \otimes Q_mU_m \otimes I_{R}) \psi^{C_MR}
(Q_1U_1 \otimes Q_2U_2  \otimes \ldots \otimes Q_mU_m \otimes I_{R})^{\dag},
\end{equation}
where $U:=U_1 \otimes U_2 \otimes \ldots \otimes U_m$. Then, we have
\begin{equation}\label{eq:decoupling}
\displaystyle
\int_{\mathbb{U}(C_1)}\int_{\mathbb{U}(C_2)}\cdots\int_{\mathbb{U}(C_m)}
\bigg \| \omeg(U)  - \decouple \bigg \|_1 dU \leq
\frac{L}{d_{C_M}} \sqrt{d_{R} \sum_{\substack{{\cal T} \subseteq
\{1,2,\ldots,m\}\\ {\cal T} \neq \emptyset}}  \prod_{i \in {\cal
T}} L_i \Tr \bigg [ (\psi^{R{\cal T}})^2 \bigg ]},
\end{equation}
where $dU:=dU_1dU_2\ldots dU_m, \int dU_i=1$ and $L := \prod_i
L_i$.
\end{Lemma}

\begin{proof} For the remainder of this proof, we will write $\int f(U) dU$ as $\mathbb{E}[f(U)]$,
indicating expectation, and abbreviate $\omeg(U)$ by $\omeg$. We have
\begin{equation}
  \begin{split} \label{variance:eq}
 \mathbb{E} \bigg [ \bigg \| \omeg - \decouple \bigg \|^2_2 \bigg ] & =  \mathbb{E}\bigg [\bigg \| \omeg -\mathbb{E}[{\omeg}] \bigg \|^2_2 \bigg ] \\
  & =  \mathbb{E} \bigg [\Tr[\omegsquare]\bigg ] - \Tr\biggl [\mathbb{E}[{\omeg}]^2 \biggr ]. \\
 \end{split}
\end{equation}
To evaluate the average of $\Tr[\omegsquare]$ , we use the following property:
\begin{equation} \Tr[\omegsquare] = \Tr \bigg (
(\omeg \otimes \omegapos )(F^{C^1_MC'^1_M} \otimes F^{RR'}) \bigg ),
\end{equation}
where $F^{C^1_MC'^1_M} := \bigotimes_{i=1}^m F^{C^1_iC'^1_i}$ is a tensor product of swap operators \mbox{$F^{C^1_iC'^1_i} := Q_i \otimes Q_i F^{C_iC'_i} Q_i \otimes Q_i$} exchanging the system $C^1_i$ and a copied version $C'^1_i$. The expectation of $\Tr [\omegsquare]$ then becomes equal to
\newcommand {\swapC} {F^{C^1_MC'^1_M}}
\newcommand \swapR {F^{RR'}}
\newcommand \UU {UU_{C_MC'_M}}
\begin{equation}
  \begin{split}
 \mathbb{E}\bigg [&\Tr[\omegsquare]\bigg ] \\ &= \mathbb{E}\bigg [ \Tr \bigg ((\omeg \otimes \omegapos )(\swapC \otimes \swapR) \bigg )\bigg ] \\
&= \mathbb{E} \bigg [ \Tr \bigg ( (U_{C_M} \otimes U_{C'_M}
\otimes I_{RR'}) ( \psi^{C_MR} \otimes
\psi^{C'_MR'}) (U_{C_M} \otimes U_{C'_M} \otimes I_{RR'})^{\dag}(\swapC \otimes \swapR) \bigg )\bigg ] \\
&= \Tr \bigg ( (\psi^{C_MR} \otimes \psi^{C'_MR'})\mathbb{E} \bigg
[(U_{C_M} \otimes U_{C'_M})^{\dag} \swapC (U_{C_M} \otimes
U_{C'_M}) \bigg ] \otimes \swapR \bigg ), \label{avgomega:eq}
  \end{split}
\end{equation}
where we have used the shorthand $U_{C_M} := U_1 \otimes U_2
\otimes \ldots U_m$, with $U_i$ being a Haar distributed unitary
on $C_i$. The unitary $U_{C'_M}$ is identical to $U_{C_M}$ but
acts on $C'_M$. Observe that the projections $\{Q_i\}$ from the
state $\omega^{C^1_MR}$ are absorbed by the swap operators
$\{F^{C^1_iC'^1_i}\}$ (i.e $F^{C^1_iC'^1_i} = Q_i \otimes Q_i
F^{C^1_iC'^1_i} Q_i \otimes Q_i$). The expectation $\mathbb{E}
\bigg [(U_{C_M} \otimes U_{C'_M})^{\dag} \swapC(U_{C_M} \otimes
U_{C'_M})\bigg ]$  can then be expanded as
\begin{equation} \label{eq:expand1}
 \mathbb{E} \bigg [(U_{C_M} \otimes U_{C'_M})^{\dag} (\swapC) (U_{C_M} \otimes U_{C'_M}) \bigg ] =
   \bigotimes^m_{i=1} \mathbb{E} \bigg [(U^{\otimes 2}_i)^{\dag} F^{C^1_iC'^1_i} U^{\otimes 2}_i \bigg ], \\
\end{equation}
where we have used the shorthand $U^{\otimes 2}_i := U_i \otimes
U_i$. Each of the expected values $\mathbb{E} \bigg [(U^{\otimes
2}_i)^{\dag} F^{C^1_iC'^1_i} U^{\otimes 2}_i\bigg ]$ can be
re-expressed, using an argumentation similar to the one found in
Appendix B of \cite{Merge}, as
\begin{equation} \label{eq:expand2}
\mathbb{E} \bigg [(U^{\otimes 2}_i)^{\dag} F^{C^1_iC'^1_i}
U^{\otimes 2}_i \bigg ] =r_i I^{C_iC'_i} + s_i
F^{C_iC'_i}, \\
\end{equation} where the coefficients $r_i$ and $s_i$ are defined as
\begin{equation}
 \begin{split}\label{coeffs:eq}
   r_i &:= \frac{L_i}{d_{C_i}}  \frac{d_{C_i} - L_i}{d^2_{C_i} - 1} \leq \frac{L_i}{d^2_{C_i}}, \\
 s_i &:= \frac{L_i^2}{d_{C_i}} \frac{d_{C_i} - 1}{d^2_{C_i} -1} \leq
\frac{(L_i)^2}{d^2_{C_i}}. \\
\end{split}
\end{equation} Substituting eqs.~(\ref{eq:expand1}), (\ref{eq:expand2}) and (\ref{coeffs:eq}) into eq.~(\ref{avgomega:eq}), we get
\begin{equation}
  \begin{split}\label{set:eq}
    \mathbb{E} \bigg [\Tr \omegsquare \bigg ] &= \Tr \bigg [ ( \psi^{C_MR} \otimes \psi^{C'_MR'}) \bigotimes_{i=1}^m
    \bigg (r_i I^{C_iC'_i} + s_i F^{C_iC'_i} \bigg ) \otimes F^{RR'} \bigg ] \\
   &= \sum_{{\cal T} \subseteq \{1,2,\ldots,m\}} \prod_{i \notin {\cal T}} r_i \prod_{i \in {\cal T}} s_i \Tr \bigg
    [ (\psi^{R{\cal T}})^2 \bigg ], \\
  \end{split}
\end{equation}
where ${\cal T}$ appearing in $\psi^{R{\cal T}}$ denotes the
system $\otimes_{i \in {\cal T}}C_i$. When ${\cal T}$ is the empty
set, the last expression in eq.~(\ref{set:eq}) reduces to
$\prod_{i=1}^m r_i \Tr[(\psi^{R})^2]$. From eq.~(\ref{coeffs:eq}),
we can bound the quantity $\prod_{i=1}^m r_i \Tr[(\psi^{R})^2]$
from above by:
\begin{equation}
\begin{split}
\prod_{i=1}^m r_i \Tr[(\psi^R)^2] &\leq \frac{L}{d^2_{C_M}} \Tr[(\psi^{R})^2] \\
& = \Tr \bigg [\frac{L^2}{d^2_{C_M}} (\tau^{C^1_M})^2 \otimes (\psi^{R})^2 \bigg ] \\
& = \Tr \bigg [ \bigg ( \frac{L}{d_{C_M}} \tau^{C^1_M} \otimes \psi^{R} \bigg )^2 \bigg ] \\
& = \Tr \biggl [ \mathbb{E} [ \omeg ]^2\biggr ]. \\
\end{split}
\end{equation}
Hence, using eqs.~(16), (20), (21) and the previous bound, we have
\begin{equation}\label{eq:decouple2}
  \mathbb{E} \bigg [ \bigg \| \omeg - \decouple \bigg \|^2_2 \bigg ] \leq \sum_{\substack{{\cal T} \subseteq \{1,2,\ldots,m\}\\ {\cal T} \notin \emptyset}} \prod_{i \notin {\cal T}} \frac{L_i}{d^2_{C_i}} \prod_{i \in {\cal T}} \frac{(L_i)^2}{d^2_{C_i}} \Tr
  \bigg [ (\psi^{R{\cal T}})^2 \bigg ].
\end{equation}
To obtain a bound on $\mathbb{E} \bigg [ \bigg \| \omeg  -
\decouple \bigg \|_1 \bigg ]$, we use the Cauchy-Schwarz
inequality:
\begin{equation}
  \begin{split}
    \mathbb{E} \bigg [ \bigg \| \omeg  - \decouple \bigg \|^2_1 \bigg ] & \leq L d_{R} \mathbb{E}\bigg [ { \bigg \| \omeg - \decouple \bigg \|^2_2}\bigg ] \\ & \leq L d_{R} \sum_{\substack{{\cal T} \subseteq \{1,2,\ldots,m\}\\ {\cal T} \neq \emptyset}} \prod_{i \notin {\cal T}} \frac{L_i}{d^2_{C_i}} \prod_{i \in {\cal T}} \frac{(L_i)^2}{d^2_{C_i}} \Tr \bigg [ (\psi^{R{\cal T}})^2 \bigg ] \\ & \leq L^2 \frac{d_{R}}{d^2_{C_M}} \sum_{\substack{{\cal T} \subseteq \{1,2,\ldots,m\}\\ {\cal T} \neq \emptyset}}  \prod_{i \in {\cal T}} L_i  \Tr \bigg [ (\psi^{R{\cal T}})^2 \bigg ]. \\
   \end{split}
\end{equation} And thus,
\begin{equation}
  \mathbb{E} \bigg [ \bigg \| \omeg  - \decouple \bigg \|_1 \bigg ] \leq \frac{L}{d_{C_M}}
\sqrt{d_{R} \sum_{\substack{{\cal T} \subseteq \{1,2,\ldots,m\}\\
{\cal T} \neq \emptyset}}  \prod_{i \in {\cal T}} L_i \Tr \bigg [
(\psi^{R{\cal T}})^2 \bigg ]}.
\end{equation}
\end{proof}

\begin{proof+}{of Proposition \ref{prop:isometry}} Fix a random measurement by choosing, for each sender $C_i$, $N_i:= \lfloor \frac{d_{C_i}K_i}{L_i} \rfloor$ fixed orthogonal
subspaces of $C_iC^0_i$ of dimension $L_i$ and one of dimension
$L'_i = d_{C_i}K_i - N_iL_i < L_i$. The projectors onto these
subspaces followed by a fixed unitary mapping it to $C^1_i$, we
denote by $Q^j_i$, $j=0,...,N_i$. Note that $Q^0_i$ projects onto a subspace of dimension $L'_i < L_i$.
Set $P^j_i := Q^j_i U_i$ with a Haar distributed random
unitary $U_i$ on $C_iC^0_i$. Applying Lemma \ref{Lemma:rdmisometry} for a measurement outcome $J=(j_1,j_2,\ldots, j_m)$,
with $\omeg_{J} = (Q^{j_1}_1U_1 \otimes Q^{j_2}_2U_2 \otimes \ldots
\otimes Q^{j_m}_mU_m \otimes I_{R})\psi^{C_MR}\otimes
\tau^{C^0_M}(Q^{j_1}_1U_1 \otimes Q^{j_2}_2U_2 \otimes \ldots \otimes
Q^{j_m}_mU_m \otimes I_{R})^{\dag}$, we have
\begin{equation}
  \begin{split} \mathbb{E} \bigg [\sum_{j_{1}=1}^{N_1} \sum_{j_2=1}^{N_2} \cdots \sum_{j_m=1}^{N_m} &\bigg \| \omeg_{J}
- \decoupled \bigg \|_1 \bigg ] \\ &\leq \bigg ( \prod_{i=1}^m N_i
\bigg ) \frac{L}{d_{C_M}K} \sqrt{d_{R} \sum_{\substack{{\cal T}
\subseteq \{1,2,...,m\} \\ {\cal T} \neq \emptyset}}  \prod_{i \in
{\cal T}} \frac{L_i}{K_i}  \Tr
\bigg [ (\psi^{R{\cal T}})^2 \bigg ]} \\
&\leq \sqrt{d_{R} \sum_{\substack{{\cal T} \subseteq \{1,2,...,m\}
\\ {\cal T} \neq \emptyset}}  \prod_{i \in {\cal T}}
\frac{L_i}{K_i}
\Tr \bigg [ (\psi^{R{\cal T}})^2 \bigg ]}. \\
  \end{split}\label{eq:measurement}
\end{equation} Taking the normalisation into account, with $p_{J}=\Tr(\omega^{C^1_MR}_{J})$ and
$\psi^{C^1_MR}_J=\frac{1}{p_J}\omega^{C^1_MR}_J$, we need to show
that on average, the $p_J$ are close to $\frac{L}{d_{C_M}K}$.
Looking at eq.~(\ref{eq:measurement}) and tracing out, we get

\begin{equation}
  \mathbb{E} \bigg [\sum_{j_{1}=1}^{N_1} \sum_{j_2=1}^{N_2} \cdots \sum_{j_m=1}^{N_m} \bigg | p_{J} -
\frac{L}{d_{C_M}K} \bigg | \bigg ] \leq \sqrt{d_{R}
\sum_{\substack{{\cal T} \subseteq \{1,2,\ldots,m\}
\\ {\cal T} \neq \emptyset}}  \prod_{i \in {\cal T}} \frac{L_i}{K_i} \Tr \bigg [ (\psi^{R{\cal T}})^2 \bigg ]}. \\
\end{equation} Hence we obtain, using the triangle inequality,

\begin{equation}
 \mathbb{E} \bigg [\sum_{j_{1}=1}^{N_1} \sum_{j_2=1}^{N_2} \cdots \sum_{j_m=1}^{N_m} p_{J} \bigg \|
\psi^{C^1_MR}_J - \tau^{C^1_M} \otimes \psi^{R} \bigg \|_1 \bigg ]
\leq 2\sqrt{d_{R} \sum_{\substack{{\cal T} \subseteq
\{1,2,\ldots,m\}
\\ {\cal T} \neq \emptyset}}  \prod_{i \in {\cal T}}
\frac{L_i}{K_i} \Tr \bigg [ (\psi^{R{\cal T}})^2 \bigg ]} =:
\Gamma_{\psi \otimes \Phi^K}.
\end{equation}

Lastly, we need to consider what happens when at least one sender
$i$ obtains a measurement outcome $j_i$ equal to 0. For an outcome
$J=(j_1,j_2,\ldots,j_m)$, define the subset ${\cal T}(J) \subseteq
\{1,2,...m\}$ such that $i \in {\cal T}(J)$ iff $j_i = 0$. Also,
define the set ${\cal Z} = \{J:|{\cal T}(J)| > 0\}$. Then, it is
easy to show that the cardinality of the set ${\cal Z}$ is
 \begin{equation}
 |{\cal Z}| = \sum_{\substack{{\cal T} \subseteq \{1,2,...,m\} \\ {\cal T} \neq \emptyset}} \prod_{i \notin {\cal T}} N_i.
 \end{equation}
For an outcome $J \in {\cal Z}$, the expected probability of the state $\omega^{C^1_MR}_{J}$ is given
by
\begin{equation}
 \begin{split}
 \mathbb{E}_{U_1U_2\ldots U_m}\bigg [\Tr(\omega^{C^1_MR}_J)\bigg ] &= \Tr \bigg [ \mathbb{E}_{U_1U_2\ldots U_m}( \omega^{C^1_MR}_J) \bigg ]\\
 &= \Tr \bigg [ \bigotimes_{i \in {\cal T}(J)} Q^0_i \tau^{C_iC^0_i} (Q^0_i)^{\dag} \bigotimes_{i \notin {\cal T}(J)} Q^{j_i}_i \tau^{C_iC^0_i} (Q^{j_i}_i)^{\dag} \bigg ] \\
 &=\frac{\prod_{i \in {\cal T}(J)} L_i' \prod_{i \notin {\cal T}(J)} L_i}{d_{C_M}K}.\\
 \end{split}
\end{equation}
 With this formula in hand and the fact that the trace norm between two states is at most 2, we
can bound the expected value of the quantum error $Q_{\cal I}(\initstate \otimes \Phi^K)$ as follows:

\begin{equation} \label{eq:final}
  \begin{split}
   \mathbb{E} \bigg [\sum_{j_{1}=0}^{N_1} \sum_{j_2=0}^{N_2} \cdots \sum_{j_m=0}^{N_m} p_{J} \bigg \| \psi^{C^1_MR}_J - \tau^{C^1_M} \otimes \psi^R \bigg \|_1 \bigg ] &\leq 2\sum_{\substack{{{\cal T} \subseteq \{1,2,...,m\}} \\ {\cal T} \neq \emptyset}} \frac{\prod_{i \in {\cal T}} L_i' \prod_{i \notin {\cal T}} N_iL_i}{d_{C_M}K} + \Gamma_{\psi \otimes \Phi^K}\\
&\leq 2 \sum_{\substack{{\cal T} \subseteq \{1,2,...,m\} \\ {\cal T} \neq \emptyset}}\prod_{i \in {\cal T}} \frac{L_i'}{d_{C_i}K_i} + \Gamma_{\psi \otimes \Phi^K} \\
&\leq 2 \sum_{\substack{{\cal T} \subseteq \{1,2,...,m\} \\ {\cal T} \neq \emptyset}}\prod_{i \in {\cal T}} \frac{L_i}{d_{C_i}K_i} + \Gamma_{\psi \otimes \Phi^K}. \\
  \end{split}
\end{equation}
\end{proof+}

\section{Multiparty State Merging: I.I.D Version} \label{sec:iid}

In this section, we analyze the case where the parties have at
their disposal arbitrarily many copies of the state $\initstate$. We give a proof of Theorem \ref{thm:statemerging}, and
then look at the case of distributed compression as an application. As mentioned earlier, the rates characterized by eq.~(\ref{eq:Mainthm3}) will be achievable without the need for a time-sharing strategy. Indeed, we will show the existence of multiparty merging protocols where each sender performs a single measurement on his share of the input state and communicates the outcome to the receiver. If the parties were to employ a time-sharing strategy, on the other hand, the many initial copies of the input state $\initstate$ would need to be divided into blocks, and for each of these blocks, the senders would have to perform a different measurement.

\subsection{Proof of Theorem \ref{thm:statemerging}}
\begin{proof} To prove the direct statement of the theorem, we use Proposition \ref{prop:isometry} in combination with Shumacher compression \cite{Schumacher}.
For $n$ copies of the state $\initstate$, consider the Schumacher
compressed state \be \ket{\Omega} := (\Pi_{\ti{B}} \otimes
\Pi_{\ti{C}_1} \otimes \Pi_{\ti{C}_2} \otimes \ldots \otimes
\Pi_{\ti{C}_m} \otimes \Pi_{\ti{R}}) \ket{\psi}^{\otimes n},\ee
and its normalized version $\ket{\Psi} := \frac{1}{\sqrt{\langle
\Omega | \Omega \rangle}}\ket{\Omega}$. Here, the systems
$\ti{B},\ti{C}_1,\ldots,\ti{C}_m,\ti{R}$ are the typical subspaces
(see \cite{Hayden001} for detailed definitions) of $B^n,C^n_1,\ldots,C^n_m,R^n$ and
$\Pi_{\ti{B}},\Pi_{\ti{C}_1},\ldots,\Pi_{\ti{C}_m},\Pi_{\ti{R}}$
are the projection operators onto these typical subspaces. In
particular, we have that
\begin{equation}
 \langle \Omega | \Omega \rangle = \bra{\psi}^{\otimes n} \Pi_{\ti{B}} \otimes \Pi_{\ti{C}_1} \otimes \ldots \otimes \Pi_{\ti{C}_m} \otimes \Pi_{\ti{R}} \ket{\psi}^{\otimes n} \geq 1-\epsilon
\end{equation} for any $\epsilon > 0$ and large enough $n$. Furthermore, we can set $\epsilon$ to be equal to $(m+2)\mathrm{exp}(-c\delta^2n)$
for some constant $c$, where $\delta > 0$ is a typicality
parameter. This follows from the fact (see Appendix~\ref{app:operatorineq}) that
 \begin{equation}\label{eq:operatorineq}
    \Pi_{\ti{B}} \otimes \Pi_{\ti{C_1}} \otimes \ldots \otimes \Pi_{\ti{C_m}} \otimes \Pi_{\ti{R}} \geq \Pi_{\ti{B}} +
   \Pi_{\ti{C_1}} + \ldots + \Pi_{\ti{C_m}} + \Pi_{\ti{R}} - (m+1) I_{\ti{B}\ti{C}_1\ldots \ti{C}_m\ti{R}} \\
 \end{equation}
  and, that by typicality, we have $\Tr(\psi^{\otimes n}_{\ti{B}}\Pi_{\ti{B}}), \Tr(\psi^{\otimes n}_{\ti{R}}\Pi_{\ti{R}}),
  \Tr(\psi^{\otimes n}_{\ti{C}_i}\Pi_{\ti{C}_i}) \geq 1 - \mathrm{exp}(-c\delta^2n)$ for all $1 \leq i \leq m$ (see \cite{Winter01}
  for the exponential bounds). Note that we have omitted some identity operator factors on the right hand side for the sake of clarity.
  The operator $\Pi_{\ti{B}}$ on the right hand side of eq.~(\ref{eq:operatorineq}) is in fact $(I^{C_MR} \otimes \Pi_{\ti{B}})$, and
  the same applies for all the other projectors on that side of the inequality.

The properties for the typical projectors
$\Pi_{\ti{B}},\Pi_{\ti{C}_1},\ldots,\Pi_{\ti{C}_m}$ allow us to
tightly bound the various dimensions and purities appearing in
Proposition \ref{prop:isometry} by appropriate "entropic"
formulas. In particular, we have \cite{Hayden001} for any system $F=C_i, B, R$:
  \begin{equation}
        \begin{split}
         (1-\epsilon)2^{n(S(F)_{\psi} - \delta)} &\leq \Tr[\Pi_{\ti{F}}] \leq 2^{n(S(F)_{\psi}+\delta)} \\
              &\Tr [ (\Psi^F)^2 ] \leq 2^{-n(S(F)_{\psi} - \delta)}. \\
         \end{split}
  \end{equation}

Hence, all parties follow a merging protocol as if they had
$\Psi$, with additional entanglement $\Phi^K:=\Phi^{K_1} \otimes
\Phi^{K_2} \otimes \ldots \otimes \Phi^{K_m}$. If each sender $C_i$
performs a random measurement on his system, as in Proposition
\ref{prop:isometry}, with projectors of rank $L_i$ (and one of
rank $L'_i \leq L_i$) such that
 \begin{equation}\label{eq:upper2}
  \prod_{i \in {\cal T}} \frac{L_i}{K_i} \leq 2^{n(S(R{\cal T})_{\psi}-S(R)_{\psi}-3\delta|{\cal T}|)}
  \end{equation} holds for all
nonempty subsets ${\cal T} \subseteq \{1,2,\ldots,m\}$, then the
expected value of the quantum error $Q_{\cal I}(\Psi \otimes
\Phi^K)$ is bounded from above by
\begin{equation} \label{eq:qerror1}
 \begin{split}
\mathbb{E}_{U_1U_2\ldots U_m} [Q_{\cal I}(\Psi \otimes \Phi^K)] &\leq 2\sum_{\substack{{\cal T} \subseteq \{1,2,...,m\} \\ {\cal T} \neq \emptyset}} \prod_{i \in {\cal T}} \frac{L_i}{d_{\ti{C}_i}K_i} + 2\sqrt{d_{\ti{R}} \sum_{\substack{{\cal T} \subseteq \{1,2,\ldots,m\} \\ {\cal T} \neq \emptyset}}  \prod_{i \in {\cal T}} \frac{L_i}{K_i} \Tr \bigg [ (\Psi_{\ti{R}\ti{{\cal T}}})^2 \bigg ]} \\
&\leq 2\sum_{\substack{{\cal T} \subseteq \{1,2,...,m\} \\ {\cal T} \neq \emptyset}} \frac{2^{n(S(R{\cal T})_{\psi} -S(R)_{\psi} - \sum_{i \in {\cal T}} S(C_i)_{\psi} -2\delta|{\cal T}|)}} {(1-\epsilon)^{|{\cal T}|}} + 2 \sqrt{\sum_{\substack{{\cal T} \subseteq \{1,2,\ldots,m\} \\ {\cal T} \neq \emptyset}} 2^{-n\delta}} \\
&\leq 2\sum_{\substack{{\cal T} \subseteq \{1,2,...,m\} \\ {\cal
T} \neq \emptyset}} \frac{2^{-2n\delta|{\cal
T}|}}{(1-\epsilon)^{|{\cal T}|}} +
2\sqrt{2^{m-n\delta}-2^{-n\delta}}=\textit{O}(2^{-n\delta/2}).
\end{split}
\end{equation} To bound the first term on the right hand side of eq.~(\ref{eq:qerror1}), we have used subadditivity twice: $S(R{\cal T})_{\psi} \leq S(R)_{\psi} + S({\cal T})_{\psi}$ and $S({\cal T})_{\psi} \leq \sum_{i \in {\cal T}} S(C_i)_{\psi}$. Hence, by
Proposition \ref{prop:mergeCond}, we can conclude that there
exists a merging protocol with error $\textit{O}(2^{-n \delta
/4})$ and entanglement cost $\overrightarrow{E}:=(\log K_1 - \log
L_1,\log K_2 - \log L_2,\ldots, \log K_m -\log L_m)$. From
eq.~(\ref{eq:upper2}), the entanglement rate
$\frac{1}{n}\overrightarrow{E}$ must satisfy
 \begin{equation}\label{eq:upper3}
  \begin{split}
  \sum_{i \in {\cal T}} \frac{1}{n} (\log K_i - \log L_i) &\geq (S(R)_{\psi} - S(R{\cal T})_{\psi} + 3\delta|{\cal T}|) \\
  & = S({\cal T}|{\overline{\cal T}}B)_{\psi} + 3\delta|{\cal T}|\\
  \end{split}
  \end{equation}
for all non empty subsets ${\cal T} \subseteq
\{1,2,\ldots,m\}$. Since we have a vanishing error for this
protocol as $n$ goes to infinity, all rate-tuples
$\overrightarrow{R}$ satisfying the preceding set of inequalities
are achievable for merging for the state $\Psi$. However, by the
gentle measurement lemma and the triangle inequality,
 \begin{equation} \bigg \| (\initstate)^{\otimes n} - \Psi \bigg \|_1 \leq 4 \sqrt{\epsilon},
 \end{equation}
and so, if we apply the same merging protocol on the state $(\initstate)^{\otimes n} \otimes \Phi^K$, we get an error of $\textit{O}(2^{-n \delta /4})+ \textit{O}(2^{-cn\delta^2/2})$. This error also vanishes as $n$ goes to infinity, and since $\delta$ was arbitrarily chosen, we can conclude that any rate-tuple $\overrightarrow{R}=(R_1,R_2,\ldots,R_m)$ satisfying
 \be \label{eq:thm3}
\sum_{i \in {\cal T}} R_i \geq S({\cal T}|\overline{{\cal T}}B)_{\psi}
 \ee
 for all non empty ${\cal T} \subseteq \{1,2,\ldots,m\}$ must be contained in the rate region. This proves the direct part of Theorem \ref{thm:statemerging}.

The converse was established in \cite{Merge} so we are done.

\end{proof}

\subsection{Distributed compression for two senders}
To illustrate some of the properties of our protocol, let's
consider the problem of distributed compression for two senders
sharing a state $\psi^{C_1C_2R}$, with purifying system $R$. The
objective is the same as in state merging, except that the decoder
has no prior information about the state. Quantum communication
will be achieved using pre-shared EPR pairs and classical
communication. The decoder will recover the original state by
applying an appropriate decoding operation. Pure entanglement, shared between the decoder and
the involved senders, might also be distilled in the process. If we
let $R_i$ denote the net amount of entanglement consumed (or
generated if $R_i$ is negative) in a distributed compression
scheme, it was found in \cite{Merge} that the rates must obey
\begin{equation} \label{eq:condDistCompress}
  \begin{split}
   R_1 &\geq S(C_1 | C_2)_{\psi} \\
   R_2 &\geq S(C_2 | C_1)_{\psi} \\
   R_1 + R_2 &\geq S(C_1C_2)_{\psi}. \\
   \end{split}
\end{equation}

\begin{figure}[t!]
\centering
    \includegraphics[width=0.7\textwidth]{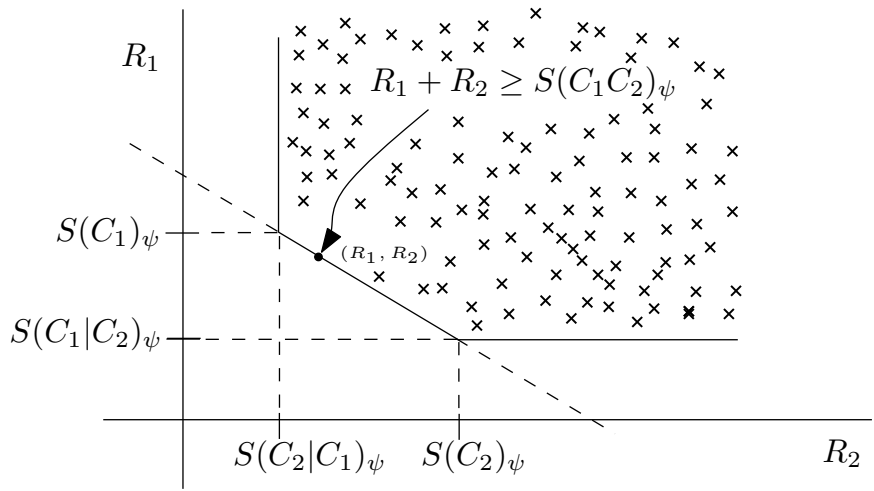}
\caption{The rate region for distributed compression ($m=2$) when
the conditional entropies $S(C_1|C_2)_{\psi}$ and
$S(C_2|C_1)_{\psi}$ are both positives. For the point $(R_1,R_2)$
on the boundary, time-sharing is needed if we perform two
applications of the original state merging protocol. Our protocol, on the other hand, can achieve this rate
without the need for time-sharing. }\label{fig:rateregion}
\end{figure}

Observe that this is just a special case of eq.~(\ref{eq:thm3})
with $m=2$ and Bob having no side information. Figure
\ref{fig:rateregion} shows the achievable rate region when the
conditional entropies have positive values.

One way to perform distributed compression is to apply the
original state merging protocol as many times as needed, adjusting
the amount of pre-shared entanglement required depending on the
information the decoder has after each application of the
protocol. For instance, if we wish to first transfer $C_1$'s share
to Bob, we can apply the state merging protocol using an
entanglement rate of $S(C_1)_{\psi}$, which amounts to Schumacher
compressing the state $\psi^{C_1}$ since the receiver has no prior
information about the state $\psi^{C_1C_2R}$. Then, to transfer
$C_2$'s share of the state, we perform another state merging, this
time, with an entanglement rate of $S(C_2|C_1)_{\psi}$. This will
correspond to one specific corner of the boundary in Figure
\ref{fig:rateregion}. Transferring $C_2$ first, and then $C_1$
will give us the other corner. To attain all other points on the
boundary using this approach, time-sharing will be required.

The techniques used to prove Theorem \ref{thm:statemerging}, however, demonstrated that time sharing is not essential to the task
of multiparty state merging. Let $(R_1,R_2)$ be any point in the
rate region. Then, $R_1$ and $R_2$ must satisfy
eq.~(\ref{eq:condDistCompress}), and so by Theorem
$\ref{thm:statemerging}$, the rate-tuple $(R_1,R_2)$ is achievable
for multiparty merging for the state $\psi^{C_1C_2R}$. That is,
given a large number of copies of $\psi^{C_1C_2R}$, there exist
multiparty state merging protocols $\Lambda^2_{\rightarrow}$ of vanishing error and
entanglement rate $\frac{1}{n}(\log K^n_1 - \log L^n_1, \log K^n_2 -
\log L^n_2)$ approaching $R_1$ and $R_2$ respectively. In the proof of Theorem \ref{thm:statemerging}, we have shown the existence of merging protocols of a specific kind. For these protocols, each sender performs a single measurement with projectors of rank $L_i$ (and one of rank $L_i' \leq L_i$) on his share $(C_iC^0_i)^{\otimes n}$ . The amount of pre-shared entanglement required and the rank of the projectors will need to satisfy eq.~(\ref{eq:upper2}). The receiver will then apply a decoding $U_J$ once he receives the outcome of the measurements. These protocols do not partition the input state $(\psi^{C_1C_2R})^{\otimes n}$ to achieve the desired rates $(R_1, R_2)$. Hence, time-sharing is not required and the parties can perform merging at any rate $(R_1, R_2)$ lying in the rate region if they were supplied with enough initial entanglement.

\section{Min-Entropies and One-Shot Merging} \label{sec:one-shot}

\subsection{Review of Min- and Max-Entropies}
Quantum min- and max-entropies are adaptations of the classical
R\'enyi entropies of order $\alpha$ when $\alpha \rightarrow
\infty$ and $\alpha=1/2$ respectively. The R\'enyi entropies were
introduced by R\'enyi \cite{Renyi} in 1961 as alternatives to the
Shannon entropy as measures of information. Although introduced in
an operational way, the Shannon entropy can also be regarded as
the unique function which satisfies a set of prescribed
postulates. R\'enyi showed that by generalizing some of the
postulates, other information-theoretic quantities could be
obtained, and this gave rise to the family of R\'enyi entropies,
parameterized by a positive number $\alpha$. R\'enyi entropies and
their quantum generalizations have found applications in areas
such as cryptography
\cite{Cachin97smoothentropy,Renner05simpleand} and statistics
\cite{Order,Diversity}. For our purposes, only the definitions and
some basic properties of the min- and max- entropies will actually
be needed.

Let ${\cal S}_{\leq}(AR)$ be the set of sub-normalized density operators (i.e $\Tr (\bar{\rho}^{AR}) \leq 1$) on the space $AR$.
The quantum min-entropy \cite{Renner02} of an
operator $\rho^{AR} \in {\cal S}_{\leq}(AR)$ relative to a density operator $\sigma^{R}$ is given by
 \begin{equation}
     H_{\min}(\rho^{AR}|\sigma^{R}):= -\log \lambda,
 \end{equation}
where $\lambda$ is the minimum real number such that $\lambda (I^A \otimes \sigma^{R}) - \rho^{AR}$ is
positive semidefinite. The conditional min-entropy $H_{\min}(\rho^{AR}|R)$ is obtained by maximizing the previous quantity over
all density operators $\sigma^{R}$:
 \begin{equation}
    H_{\min}(\rho^{AR}|R) := \sup_{\sigma^{R}} H_{\min}(\rho^{AR}|\sigma^{R}).
  \end{equation}
 For two sub-normalized states $\rho$ and $\bar{\rho}$, we define the purified distance between $\rho$ and $\bar{\rho}$ as
 \begin{equation}
    P(\rho,\bar{\rho}) := \sqrt{1- \overline{F}(\rho,\bar{\rho})^2},
 \end{equation}
 where $\overline{F}(\rho,\bar{\rho})$ is the generalized fidelity between $\rho$ and $\bar{\rho}$ (see \cite{Renner01} for the definition). It is related to the trace distance $D(\rho, \bar{\rho}) := \frac{1}{2}\| \rho - \bar{\rho} \|_1$ as follows
 \begin{equation}
    D(\rho, \bar{\rho}) \leq P(\rho, \bar{\rho}) \leq 2 \sqrt{D(\rho,\bar{\rho})}.\label{eq:purified}
 \end{equation}
A proof of this fact can be found in Lemma 6 of \cite{Renner01}. (Lemma 6 actually relates the purified distance to the generalized distance $\bar{D}(\rho, \bar{\rho})$. However, $\bar{D}(\rho,\bar{\rho})$ is always greater than or equal to the trace distance and bounded above by $\|\rho - \bar{\rho}\|_1$.)

Using the purified distance as our measure of closeness, we obtain a family of smooth min-entropies $\{H^{\epsilon}_{\min}\}$ by optimizing over all sub-normalized density operators close to $\rho^{AB}$ with respect to $P(\bar{\rho},\rho)$:
 \begin{equation}
    H^{\epsilon}_{\min}(\rho^{AR}|R) := \sup_{\bar{\rho}^{AR}} H_{\min}(\bar{\rho}^{AR}|R),
 \end{equation}
where the supremum is taken over all $\bar{\rho}^{AR}$ such that $P(\bar{\rho}^{AR},\rho^{AR}) \leq \epsilon$. If we use the trace distance instead as our measure of closeness, we obtain the family $\{\bar{H}_{\min}^{\epsilon}\}$:
 \begin{equation}
    \bar{H}^{\epsilon}_{\min}(\rho^{AR}|R) := \sup_{\bar{\rho}^{AR}} H_{\min}(\bar{\rho}^{AR}|R),
 \end{equation}
where the supremum is taken over all sub-normalized $\bar{\rho}^{AR}$ such that $D(\bar{\rho}^{AR}, \rho^{AR}) \leq \epsilon$. From eq.~(\ref{eq:purified}), the smooth min-entropy $H^{\epsilon}_{\min}$ can be bounded by $\bar{H}^{\epsilon}_{\min}$ in the following way:
 \begin{equation}
    H^{\epsilon}_{\min}(\rho^{AR}|R) \leq \bar{H}^{\epsilon}_{\min}(\rho^{AR}|R) \leq H^{2\sqrt{\epsilon}}_{\min}(\rho^{AR}|R). \label{eq:minIneqs}
 \end{equation}

Given a purification $\rho^{ABR}$ of $\rho^{AR}$, with purifying
system $B$, the family of smooth max entropies $\{H^{\epsilon}_{\max}\}$ is defined as
 \begin{equation}
    H^{\epsilon}_{\max}(\rho^{AB}|B) := - H^{\epsilon}_{\min}(\bar{\rho}^{AR}|R)
 \end{equation}
for any $\epsilon \geq 0$. The smooth max-entropies can also be expressed as
 \begin{equation}\label{eq:smoothmax}
    H^{\epsilon}_{\max}(\rho^{AB}|B) = \inf_{\bar{\rho}^{AB}} H_{\max}(\bar{\rho}^{AB}|B),
 \end{equation}
where the infimum is taken over all $\bar{\rho}^{AB}$ such that $P(\bar{\rho},\rho) \leq \epsilon$. We refer to \cite{Renner01} for a proof of this fact. When $\epsilon=0$, an alternative expression for the max-entropy $H_{\max}(\psi^{AB}|B)$ was obtained in \cite{Renner03}:
\begin{equation}
  H_{\max}(\rho^{AB}|B) = \sup_{\sigma^B} \log F^2(\rho^{AB}, I^A \otimes \sigma^B),
\end{equation}
where the supremum is taken over all density operators $\sigma^B$ on the space $B$. From this last expression, the smooth max-entropy $H^{\epsilon}_{\max}(\rho^{A})$ of a sub-normalized operator $\rho^A \in {\cal S}_{\leq}(A)$ reduces to
\begin{equation}
  H^{\epsilon}_{\max}(\rho^{A}) = 2\log \sum_x \sqrt{\bar{r}_x},
\end{equation}
where $\bar{r}_x$ are the eigenvalues of the sub-normalized density operator $\bar{\rho}^A$ which optimizes the right hand side of eq.~(\ref{eq:smoothmax}).

When defining the smooth min- and max-entropies using the purified
distance, other useful properties such as quantum data processing
inequalities and concavity of the max-entropy are also known to hold.
A detailed analysis of these properties can be found in \cite{Renner01}.

We will also need, for technical reasons, another entropic
quantity called the conditional collision entropy
$H_2(\rho^{AB}|\sigma^{B})$\cite{Renner01}. It is defined as
\begin{equation}
   H_2(\rho^{AB}|\sigma^{B}):=  -\log \Tr \bigg [ \bigg ((I_A \otimes \sigma_{B}^{-1/4}) \rho^{AB} (I_A \otimes \sigma_B^{-1/4}) \bigg )^2 \bigg ].
\end{equation} It is a quantum adaptation of the classical conditional collision entropy. We have the following lemma relating the min-entropy to the collision entropy:
\begin{Lemma}\cite{Renner01}
For density operators $\rho^{AB}$ and $\sigma^{B}$ with $\mathrm{supp}\{\Tr_A(\rho^{AB})\} \subseteq \mathrm{supp}\{\sigma^{B}\}$, we have
  \[ H_{\min}(\rho^{AB}|\sigma^{B}) \leq H_2(\rho^{AB}|\sigma^{B}). \]
\end{Lemma}

The last two results we will need are the additivity of the min-entropy and the following lemma which relates the trace norm of an
hermitian operator $S$ to its Hilbert-Schmidt norm, with respect
to a positive semidefinite operator $\sigma$:
 \begin{Lemma}\label{Lemma:TrSchmidt}
   Let $S$ be an hermitian operator acting on some space $X$ and $\sigma$ be a positive semidefinite operator on $X$. Then, we have
     \begin{equation}
         \|S\|_1 \leq \sqrt{\Tr(\sigma)}  \|\sigma^{-1/4} S \sigma^{-1/4} \|_2.
     \end{equation}
 \end{Lemma}

\begin{Lemma}[Additivity]
Let $\rho^{AB}$ and $\rho^{A'B'}$ be sub-normalized operators on the spaces $AB$ and $A'B'$ respectively. For density operators $\sigma^{B}$ and $\sigma^{B'}$, we have
\begin{equation}
  H_{\min}(\rho^{AB} \otimes \rho^{A'B'}|\sigma^{B} \otimes \sigma^{B'}) = H_{\min}(\rho^{AB}|\sigma^{B}) + H_{\min}(\rho^{A'B'}|\sigma^{B'}).
\end{equation}
\end{Lemma}
For proofs of the preceding two lemmas, see \cite{Berta}.

\subsection{Characterizing the entanglement cost of merging using min-entropies}

In \cite{Berta}, Berta showed that the smooth min-entropy is the
information theoretic measure which quantifies the minimal amount
of entanglement necessary for performing state merging when a
single copy of $\psi^{ABR}$ is available. More specifically, he
proved that the minimal entanglement cost $\log K - \log L$
necessary for merging the $A$ part of a state $\psi^{ABR}$ to the
location of the $B$ system is bounded from below by
 \begin{equation}\label{eq:lowerboundCost}
    \log K - \log L \geq -\bar{H}^{\sqrt{\epsilon}}_{\min}(\psi^{AR}|R).
 \end{equation}Furthermore, he demonstrated the existence of a state merging protocol using an entanglement
 cost\footnote{The numbers $K, L$ are natural numbers, and so we must choose values for $K$ and $L$ such that $\log K - \log L$ is minimal, but greater or equal than the right hand side of eq.~(\ref{eq:berta}).} of
 \begin{equation}\label{eq:berta}
   \log K - \log L = -\bar{H}_{\min}^{\frac{\epsilon^2}{64}}(\psi^{AR}|R) + 4\log \left(\frac{1}{\epsilon} \right ) +
   12.
 \end{equation}
This last result was derived by re-expressing the upper bound to
the quantum error (Lemma 6 of \cite{Merge}) as a function of the
smooth min-entropy.

In this section, we would like to generalize the main results of
\cite{Berta} to the case of multiple parties sharing a state
$\initstate$. Our first result is a reformulation of Lemma
\ref{Lemma:rdmisometry} in terms of min-entropies:
\newcommand{\cR}{{\cal R}}

\begin{Lemma}[Compare to Lemma 4.5 of \cite{Berta}] \label{Lemma:oneshotdecouple}
For each sender $C_i$, let $P_i: C_i \rightarrow C_i^1$ be a projector onto
a subspace $C_i^1$ of $C_i$ and $U_i$ be a unitary on the space $C_i$. Define the state
 \begin{equation}
   \omega^{C^1_MR} := (P_1U_1 \otimes P_2U_2 \otimes \ldots
   \otimes P_mU_m \otimes I_R) \sourcestate (P_1U_1 \otimes P_2U_2
   \otimes \ldots \otimes P_mU_m \otimes I_R)^{\dag}.
 \end{equation}
If the unitaries $U_1, U_2,\ldots, U_m$ are distributed according to
the Haar measure, then for any state $\sigma^R$ of the system $R$, we have
 \begin{equation}\label{eq:weak}
   \mathbb{E} \biggl [ \biggl \|  \omega^{C^1_MR} -
   \frac{L}{d_{C_M}} \tau^{C^1_M} \otimes \psi^{R} \biggr \|_1 \biggr ] \leq
   \frac{L}{d_{C_M}}\sqrt{\sum_{\substack{\cK \subseteq \{1,2,\ldots,m\} \\ \cK
   \neq \emptyset}}
   2^{-(H_{\mathrm{min}}(\psi^{\cK R}|\sigma^R) -
   \log L_{\cK})}},
 \end{equation}
where $L_{\cK}=\prod_{i \in \cK} L_i$.
\end{Lemma}

\begin{proof} Using Lemma \ref{Lemma:TrSchmidt}, we have, for
any state $\sigma^R$ of $R$,
\begin{equation}
\label{eq:A.6} \biggl \| \omega^{C^1_MR} -
\frac{L}{d_{C_M}}\tau^{C^1_M} \otimes \psi^{R} \biggr \|_1 \leq \sqrt{L} \biggl \|
(I^{C^1_M} \otimes (\sigma^R)^{-\frac{1}{4}})(
\omega^{C^1_MR} - \frac{L}{d_{C_M}}\tau^{C^1_M} \otimes \psi^R)(I^{C^1_M} \otimes
(\sigma^R)^{-\frac{1}{4}}) \biggr \|_2.
\end{equation}
Define
\begin{equation}
\begin{split}
 \tilde{\psi}^{C_MR} &:= (I^{C_M} \otimes (\sigma^R)^{-\frac{1}{4}})
 (\psi^{C_MR})(I^{C_M} \otimes (\sigma^R)^{-\frac{1}{4}}) \\
 \tilde{\omega}^{C^1_MR} &:= (P_1U_1 \otimes P_2U_2 \otimes \ldots
 \otimes P_mU_m \otimes I_R) \tilde{\psi}^{C_MR} (P_1U_1 \otimes
 P_2U_2 \otimes \ldots \otimes P_mU_m \otimes I_R)^{\dag}. \\
 \end{split}
\end{equation}
Then, the right hand side of eq.~(\ref{eq:A.6}) can be rewritten
as $\sqrt{L}\biggl \|\tilde{\omega}^{C^1_MR} -
\frac{L}{d_{C_M}}\tau^{C^1_M} \otimes \tilde{\psi}^R \biggr \|_2$.
Using eq.~(\ref{eq:decouple2}) in the proof of Lemma
\ref{Lemma:rdmisometry}, we have

\begin{equation}\label{eq:tilde}
 \begin{split}
\mathbb{E} \biggl [ \biggl \| \tilde{\omega}^{C^1_MR} -
\frac{L}{d_{C_M}}\tau^{C^1_M} \otimes \tilde{\psi}^R \biggr \|_2^2
\biggr ] & \leq \sum_{\substack{ \cK \subseteq \{1,2,\ldots,m\} \\
\cK \neq \emptyset}} \prod_{i \notin \cK} \frac{L_i}{d^2_{C_i}}
\prod_{i \in \cK} \frac{L_i^2}{d^2_{C_i}}
\Tr \biggl [ (\tilde{\psi}^{\cK R})^2 \biggr ] \\
 & \leq \frac{L}{d^2_{C_M}} \sum_{\substack{ \cK \subseteq
\{1,2,\ldots,m\} \\ \cK \neq \emptyset}}L_{\cK} \Tr
\biggl [ (\tilde{\psi}^{\cK R})^2 \biggr ]. \\
 \end{split}
\end{equation}
The quantity $\Tr[(\tilde{\psi}^{\cK R})^2]$ can be rewritten
as:
\begin{equation}\label{eq:trace}
 \begin{split}
   \Tr [(\tilde{\psi}^{\cK R})^2] & = \Tr \biggl [ \biggl (
   \Tr_{\overline{\cK}}(\tilde{\psi}^{C_MR}) \biggr)^2 \biggr ] \\
    &= \Tr \biggl [ \biggl ( (I^{\cK} \otimes
    (\sigma^R)^{-\frac{1}{4}})(\psi^{\cK R})(I^{\cK} \otimes
    (\sigma^R)^{-\frac{1}{4}}) \biggr )^2 \biggr ] \\
   &= 2^{-H_2(\psi^{\cK R} | \sigma^R)}, \\
 \end{split}\
\end{equation}
where $H_2(\psi^{\cK R}|\sigma^R)$ is the conditional collision
entropy of $\psi^{\cK R}$ relative to $\sigma^R$. Combining
eqs.~(\ref{eq:A.6}), (\ref{eq:tilde}) and (\ref{eq:trace})
together, and using the fact that $H_{\mathrm{min}}(\psi^{\cK
R}|\sigma^{R}) \leq H_2(\psi^{\cK R}|\sigma^R)$, we get

\begin{equation}
 \begin{split}
   \mathbb{E} \biggl [ \biggl \| \omega^{C^1_MR} -
   \frac{L}{d_{C_M}}\tau^{C^1_M} \otimes \psi^{R} \biggr \|_1 \biggr ] & \leq \frac{L}{d_{C_M}}
   \sqrt{ \sum_{\substack{ \cK \subseteq
\{1,2,\ldots,m\} \\ \cK \neq \emptyset}}L_{\cK}
2^{-H_2(\psi^{\cK R}|\sigma^R)}} \\
 & \leq \frac{L}{d_{C_M}} \sqrt{\sum_{\substack{\cK \subseteq \{1,2,\ldots,m\} \\ \cK
   \neq \emptyset}}
   2^{-(H_{\mathrm{min}}(\psi^{\cK R}|\sigma^R) -
   \log L_{\cK})}}. \\
 \end{split}
\end{equation}
\end{proof}
With this result in hand, we are now ready to give an adaptation of Lemma 4.6 in \cite{Berta} for our general multiparty setting.

\begin{Theorem}[Compare to Proposition 4.7 of \cite{Berta}]\label{thm:cost1}
Let $\initstate$ be any $(m+2)$-partite pure state and fix
$\epsilon > 0$. Then, for any entanglement cost
$\overrightarrow{E} = (\log K_1 - \log L_1,\log K_2 - \log
L_2,\ldots, \log K_m - \log L_m)$ satisfying
 \begin{equation}\label{eq:cost1}
  \begin{split}
   \log K_{\cK} - \log L_{\cK} := \sum_{i \in \cK} (\log K_i - \log L_i) &\geq -H_{\min}(\psi^{\cK R}|\psi^R) + 4\log \left(\frac{1}{\epsilon} \right )+ 2m + 8  \\
  \end{split}
 \end{equation}
for all non-empty subsets $\cK \subseteq \{1,2,\ldots, m\}$, there
exists a state merging protocol acting on $\initstate$ with error
$\epsilon$. The set $\cKbar$ is defined as
the complement of $\cK$.
\end{Theorem}

\begin{proof}
We proceed by fixing a random measurement for each sender $C_i$ as in Proposition $\ref{prop:isometry}$. %For technical reason, we will assume that $d_{C_i}K_i = N_i L_i$ for some positive integer $N_i$.
We can describe $C_i$'s random measurement using $N_i:=\lfloor
\frac{d_{C_i}K_i}{L_i} \rfloor$ partial isometries $P^j_i =
Q^j_iU_i$, where $U_i$ is a Haar distributed random unitary acting
on the system $C_iC^0_i$ and $Q^j_i$ is as defined in Proposition
\ref{prop:isometry}. If $d_{C_i}K_i > N_i L_i$, there is an
additional partial isometry $P^0_i$ of rank $L'_i:=d_{C_i}K_i -
N_iL_i < L_i$. Applying the previous lemma to the
state $\psi^{C_MR} \otimes \tau^{K_1} \otimes \tau^{K_2} \otimes
\ldots \otimes \tau^{K_m}$, with $\sigma^R = \psi^R$, we get
 \begin{equation}\label{eq:deriv}
   \begin{split}
  \mathbb{E} \bigg [\sum_{j_{1}=1}^{N_1} \sum_{j_2=1}^{N_2} \cdots \sum_{j_m=1}^{N_m} \bigg \| &\omeg_{J}
- \decouple \bigg \|_1 \bigg ] \\ &\leq \frac{\prod_{i=1}^m N_i
L_i}{d_{C_M}K} \sqrt{\sum_{\substack{\cK \subseteq
\{1,2,\ldots,m\} \\ \cK \neq \emptyset}}
2^{-(H_{\mathrm{min}}(\psi^{\cK R}|\psi^R) + \log K_{\cK} -
   \log L_{\cK})}} \\
  &\leq \sqrt{\sum_{\substack{\cK \subseteq \{1,2,\ldots,m\} \\ \cK
   \neq \emptyset}}
   2^{-(H_{\mathrm{min}}(\psi^{\cK R}|\psi^R) + \log K_{\cK} -
   \log L_{\cK})}}, \\
\end{split}
\end{equation}
where $K_{\cK} = \prod_{i \in \cK} K_i$. Note that to get the first inequality, we have used the fact that $H_{\min}(\psi^{\cK R} \otimes \tau^{K_{\cK}}|\psi^{R}) = H_{\min}(\psi^{\cK R}|\psi^{R}) + \log K_{\cK}$.

Using eq.~(\ref{eq:cost1}), we can simplify the previous
inequality and obtain
 \begin{equation}
   \begin{split}
  \mathbb{E} \bigg [\sum_{j_{1}=1}^{N_1} \sum_{j_2=1}^{N_2} \cdots \sum_{j_m=1}^{N_m} \bigg \| &\omeg_{J}
- \decouple \bigg \|_1 \bigg ] \\ &\leq \sqrt{\sum_{\substack{\cK
\subseteq \{1,2,\ldots,m\} \\ \cK
   \neq \emptyset}}
   2^{-(H_{\mathrm{min}}(\psi^{\cK R}|\psi^R) + \log K_{\cK} -
   \log L_{\cK})}} \\
   &\leq \frac{\epsilon^2}{2^{\frac{m+8}{2}}} \leq \frac{\epsilon^2}{16}.\\
\end{split}
\end{equation}
Taking normalisation into account, with $p_J =
\Tr(\omega_J^{C^1_MR})$ and $\psi^{C^1_MR}_J =
\frac{1}{p_J}\omega_J^{C^1_MR}$, we can trace out the left hand
side of the previous set of inequalities, and get
\begin{equation}
\mathbb{E} \bigg [\sum_{j_{1}=1}^{N_1} \sum_{j_2=1}^{N_2} \cdots \sum_{j_m=1}^{N_m} \bigg | p_J - \frac{L}{d_{C_M}} \bigg |_1 \bigg ] \leq \frac{\epsilon^2}{16}. \\
\end{equation}
By applying the triangle inequality, we obtain
\begin{equation}
  \mathbb{E} \bigg [\sum_{j_{1}=1}^{N_1} \sum_{j_2=1}^{N_2} \cdots \sum_{j_m=1}^{N_m} p_J \bigg \| \psi^{C^1_MR}_{J}
- \tau^{C^1_M} \otimes \psi^R \bigg \|_1 \bigg ] \leq \frac{\epsilon^2}{8}.
\end{equation}
Using this, and eq.~(\ref{eq:final}) in the proof of Proposition
\ref{prop:isometry}, we can get an upper bound to the quantum
error $Q_{\cal I}(\psi^{C_1C_2\ldots C_mBR} \otimes \Phi^K)$:
\begin{equation}
\begin{split}
 \mathbb{E} \bigg [\sum_{j_{1}=0}^{N_1} & \sum_{j_2=0}^{N_2} \cdots \sum_{j_m=0}^{N_m} p_J \bigg \| \psi^{C^1_MR}_{J} - \tau^{C^1_M} \otimes \psi^R \bigg \|_1 \bigg ] \\
&\leq 2 \sum_{\substack{{\cal T} \subseteq \{1,2,...,m\} \\ {\cal
T} \neq \emptyset}}\prod_{i \in {\cal T}} \frac{L_i}{d_{C_i}K_i} +
\mathbb{E} \sum_{j_{1}=1}^{N_1} \sum_{j_2=1}^{N_2} \cdots
\sum_{j_m=1}^{N_m} p_J \bigg \| \psi^{C^1_MR}_{J}
- \tau^{C^1_M} \otimes \psi^R \bigg \|_1  \\
 &\leq  \sum_{\substack{{\cal T} \subseteq \{1,2,...,m\}}} \frac{2\epsilon^4 2^{H_{\min}(\psi^{\cK R}|\psi^R)}}{2^{2m+8}d_{C_{\cK}}} + \frac{\epsilon^2}{8} \\
 &\leq  \sum_{\substack{{\cal T} \subseteq \{1,2,...,m\}}} \frac{2\epsilon^4  2^{H_{\min}(\psi^{\cK})}}{2^{2m+8}d_{C_{\cK}}} + \frac{\epsilon^2}{8}\\
&\leq  \frac{\epsilon^4}{2^{m+7}} + \frac{\epsilon^2}{8} \leq \frac{\epsilon^2}{4}\\
   \end{split}
 \end{equation}
To get the third inequality, we have used the strong subadditivity
of the min-entropy \cite{Renner01}:
\[ H_{\min}(\psi^{\cK R}|\psi^R) \leq H_{\min}(\psi^{\cK}). \]
The last line follows from the fact that $H_{\min}(\psi^{\cK}) =
-\log \lambda_{\max}(\psi^{\cK}) \leq \log d_{C_{\cK}}$. From
Proposition \ref{prop:mergeCond}, we can conclude there exists a
state merging protocol acting on $\initstate$ with error
$\epsilon$.
\end{proof}

When $m=1$, the previous result yields a merging protocol of error $\epsilon$
and entanglement cost $\log K-\log L =
-H_{\min}(\psi^{C_1R}|\psi^R) + 4\log \left(\frac{1}{\epsilon} \right ) + 10$. Berta \cite{Berta}
showed however that the min-entropy $H_{\min}(\psi^{C_1R}|\psi^R)$
of the state $\psi^{C_1R}$ relative to $\psi^R$ can be replaced by
the min-entropy $H_{\min}(\psi^{C_1R}|R)$ of $\psi^{C_1R}$
relative to $R$. This yields a smaller entanglement cost, and we
can ask whether the right hand side of eq.~(\ref{eq:cost1}) can be
replaced by a formula involving min-entropies of this form when $m > 1$. To
allow this to work, we would need a more general version of Lemma
\ref{Lemma:oneshotdecouple}, where eq.~(\ref{eq:weak}) is replaced
by
\begin{equation}\label{eq:x1}
\mathbb{E} \biggl [ \biggl \|  \omega^{C^1_MR} -
   \frac{L}{d_{C_M}} \tau^{C^1_M} \otimes \psi^{R} \biggr \|_1 \biggr ] \leq
   \frac{L}{d_{C_M}}\sqrt{\sum_{\substack{\cK \subseteq \{1,2,\ldots,m\} \\ \cK
   \neq \emptyset}}
   2^{-(H_{\mathrm{min}}(\psi^{\cK R}|\sigma_{\cK}^R) -
   \log L_{\cK})}},
\end{equation}
and this inequality holds for $2^m - 1$ possibly different states
$\sigma_{\cK}^R$. Using this stronger form, we could set
$\sigma_{\cK}^R = \bar{\sigma}^R_{\cK}$, with $H_{\min}(\psi^{\cK
R}|\bar{\sigma}^R_{\cK})=H_{\min}(\psi^{\cK R}|R)$ and bound the
left hand side of eq.~(\ref{eq:x1}) by setting $\log K_{\cK} -
\log L_{\cK} \geq -H_{\min}(\psi^{\cK R}|R) +
4\log \left(\frac{1}{\epsilon} \right ) + 2m +8$. However, it is unclear if such a
stronger form of Lemma $\ref{Lemma:oneshotdecouple}$ can be
obtained.

Berta \cite{Berta} also showed that the previous result can be
further improved when $m=1$ by smoothing the min entropy
$H_{\min}(\psi^{C_1R}|R)$ around sub-normalized operators
$\bar{\psi}^{C_1R}$ which are $\epsilon$-close in the trace distance to the state
$\psi^{C_1R}$. It is also unclear if eq.~(\ref{eq:weak}) in
Lemma \ref{Lemma:oneshotdecouple} can be strengthened to
\begin{equation}\label{eq:x}
\mathbb{E} \biggl [ \biggl \|  \omega^{C^1_MR} -
   \frac{L}{d_{C_M}} \tau^{C^1_M} \otimes \psi^{R} \biggr \|_1 \biggr ] \leq
   \frac{L}{d_{C_M}}\sqrt{\sum_{\substack{\cK \subseteq \{1,2,\ldots,m\} \\ \cK
   \neq \emptyset}}
   2^{-(H^{\epsilon}_{\mathrm{min}}(\psi^{\cK R}|\sigma^R) -
   \log L_{\cK})}},
\end{equation}
for any fixed $\epsilon > 0$.

\begin{conjecture}\label{conj:cost1}
Let $\psi^{C_1C_2\ldots C_mBR}$ be an $(m+2)$-partite state. For
any $\epsilon > 0$, there exists a multiparty state merging of
error $\epsilon$ whenever the entanglement cost
$\overrightarrow{E}:=(\log K_1 - \log L_1, \log K_2 - \log L_2,
\ldots , \log K_m - \log L_m)$ satisfies
\begin{equation}\label{eq:cost2}
   \log K_{\cK} - \log L_{\cK} := \sum_{i \in \cK} (\log K_i - \log L_i) \geq   -H^{\epsilon}_{\min}(\psi^{\cK R}|R) +
   O(\log1 / \epsilon)+O(m)
 \end{equation}
for all non-empty subsets $\cK \subseteq \{1,2,\ldots, m\}$.
\end{conjecture}
The main difficulty in proving the conjecture is that it allows independent smoothing of each of the min-entropies. It is straightforward to modify our proof to allow smoothing using a common state for all the min-entropies, but the
monolithic nature of the protocol does not naturally permit tailoring
the smoothing state term-by-term. We can, however, give a partial characterization of the entanglement cost in terms of smooth min-entropies if we apply the single-shot state merging protocol of \cite{Berta} on one sender at a time.

\begin{Proposition}\label{prop:epsilonCost}
For a $(m+2)$-partite pure state $\psi^{C_1C_2\ldots C_mBR}$, fix an $\epsilon > 0$ and let $\pi:\{1,2,\ldots,m\}\rightarrow \{1,2,\ldots,m\}$ be any
ordering of the $m$-senders $C_1,C_2,\ldots,C_m$. Then, for any entanglement cost $\overrightarrow{E}:=(\log
K_1 - \log L_1, \log K_2 - \log L_2, \ldots , \log K_m - \log
L_m)$ satisfying
 \begin{equation} \label{eq:entanglementcost}
     \log K_{i} - \log L_{i} \geq  -H^{\frac{\epsilon^2}{64m^2}}_{\min}(\psi^{C_i R_{\pi^{-1}(i)}}|R_{\pi^{-1}(i)}) + 4 \log \left(\frac{m}{\epsilon} \right ) + 12 \quad \text{for all $1 \leq i \leq m$}, \\
 \end{equation}
where $R_i:=R \bigotimes^m_{j=i+1}C_{\pi(j)}$, there exists a multiparty state merging
protocol acting on the state $\psi^{C_1C_2\ldots C_mBR}$ with error $\epsilon$.
\end{Proposition}
\begin{proof}
Our multiparty state merging protocol for the state
$\psi^{C_1C_2\ldots C_mBR}$ will consists of sending each sender's
share to the receiver one at a time according to the
ordering $\pi$: The sender $C_{\pi(1)}$ will merge his part of the
state first, followed by $C_{\pi(2)}$, $C_{\pi(3)}$, etc. We can
view the input state $\psi^{C_1C_2\ldots C_m BR}$ as a tripartite
pure state $\psi^{C_{\pi(1)}R_1B}$, with the reference system
$R_1$ being composed of the systems $C_{\pi(2)}C_{\pi(3)}\ldots
C_{\pi(m)}R$. According to Berta \cite{Berta}, there exists a
state merging protocol of error $\epsilon/m$ and entanglement cost
\begin{equation}
 \begin{split} \log K'_1 - \log L'_1 &:=
-\bar{H}^{\frac{\epsilon^2}{64m^2}}_{\min}(\psi^{C_{\pi(1)} R_1}|R_1) +
4 \log \left(\frac{m}{\epsilon} \right ) + 12 \\
 &\leq -H^{\frac{\epsilon^2}{64m^2}}_{\min}(\psi^{C_{\pi(1)} R_1}|R_1) +
4 \log \left(\frac{m}{\epsilon} \right ) + 12 \\
 &\leq \log K_1 - \log L_1, \\
  \end{split}
\end{equation} which will produce an output
state $\rho^{C^1_{\pi(1)}B^1_{\pi(1)}BR_1}$ satisfying
\begin{equation}
  \bigg \| \rho_1^{C^1_{\pi(1)}B^1_{\pi(1)}B_{\pi(1)}B R_1} - \psi^{B_{\pi(1)}BR_1} \otimes \Phi^{L_1} \bigg \|_1 \leq \frac{\epsilon}{m},
\end{equation}
where the system $B_{\pi(1)}$ is substituted for the system
$C_{\pi(1)}$. After $C_{\pi(1)}$ has merged his share, the next
sender $C_{\pi(2)}$ will perform a random measurement on his share
of the state and send the measurement outcome to the receiver.
Suppose, for the moment, that the parties share the state
$\psi^{B_{\pi(1)}BR_1} \otimes \Phi^{L_1}$ instead of the output
state $\rho_1$. The state $\psi^{B_{\pi(1)}BR_1}$ can be viewed as
a tripartite state $\psi^{C_{\pi(2)}B_2R_2}$, with
$B_2:=B_{\pi(1)}B$ and $R_2:=C_{\pi(3)}C_{\pi(4)}\ldots
C_{\pi(m)}R$. Using Berta's result once more, we know there exists
a state merging protocol of error $\epsilon/m$ and entanglement
cost \begin{equation}
\begin{split}
 \log K'_2 - \log L'_2 &=
-\bar{H}^{\frac{\epsilon^2}{64m^2}}_{\min}(\psi^{C_{\pi(2)}R_2}|R_2) +
4\log \left(\frac{m}{\epsilon} \right )+12  \\
&\leq -H^{\frac{\epsilon^2}{64m^2}}_{\min}(\psi^{C_{\pi(2)}R_2}|R_2) +
4\log \left(\frac{m}{\epsilon} \right )+12  \\
& \leq \log K_2 - \log L_2, \\
\end{split}
\end{equation}
which produces an output
state $\rho_2^{C^1_{\pi(1)}C^1_{\pi(2)}B^1_{\pi(1)}B^1_{\pi(2)}B_{\pi(2)}B_2 R_2}$ satisfying
\begin{equation}
  \bigg \| \rho_2^{C^1_{\pi(1)}C^1_{\pi(2)}B^1_{\pi(1)}B^1_{\pi(2)}B_{\pi(2)}B_2 R_2} - \psi^{B_{\pi(2)}B_2R_2} \otimes \Phi^{L_1} \otimes \Phi^{L_2} \bigg \|_1 \leq \frac{\epsilon}{m},
\end{equation}
where the system $B_{\pi(2)}$ is substituted for the system
$C_{\pi(2)}$. If we apply the same protocol on the state
$\rho_1^{C^1_{\pi(1)}B^1_{\pi(1)}B_{\pi(1)}B R_1}$ instead, we get
an output state $\rho_3$ which satisfies
\begin{equation}
 \begin{split}
 \bigg \|& \rho_3^{C^1_{\pi(1)}C^1_{\pi(2)}B^1_{\pi(1)}B^1_{\pi(2)}B_{\pi(2)}B^2 R_2} - \psi^{B_{\pi(1)}B_{\pi(2)}B R_2} \otimes \Phi^{L_1} \otimes \Phi^{L_2} \bigg \|_1 \\ &\leq \|\rho_3 - \rho_2 \|_1 + \|\rho_2 - \psi^{B_{\pi(1)}B_{\pi(2)}B R_2} \otimes \Phi^{L_1} \otimes \Phi^{L_2} \|_1 \\
   &\leq \| \rho_1 - \psi^{B_{\pi(1)}BR_1} \otimes \Phi^{L_1} \|_1 + \|\rho_2 - \psi^{B_{\pi(1)}B_{\pi(2)}B R_2} \otimes \Phi^{L_1} \otimes \Phi^{L_2} \|_1 \\
  &\leq \frac{2\epsilon}{m}, \\
  \end{split}
\end{equation}
where we have used the triangle inequality and monotonicity of the
trace distance under quantum operations. The analysis for the
other senders $C_{\pi(3)},C_{\pi(4)},\ldots,C_{\pi(m)}$ can be
performed in a similar way, which leads to a multiparty state
merging protocol of error $\epsilon$ and entanglement cost $\overrightarrow{E}:=(\log K_1 - \log L_1, \log K_2 - \log L_2, \ldots, \log K_m - \log L_m)$ satisfying
eq.~(\ref{eq:entanglementcost}). The final output state $\rho_{m}$
satisfies
  \begin{equation}
  \bigg \| \rho_m - \psi^{B_{\pi(1)}B_{\pi(2)}\ldots B_{\pi(m)} B R} \otimes \Phi^{L_1} \otimes \Phi^{L_2} \otimes \ldots \otimes \Phi^{L_m}\bigg \|_1 \leq \epsilon.
  \end{equation}
\end{proof}

Figure \ref{fig:rateregion1} depicts the boundaries of the regions described by Theorem~\ref{thm:cost1}, Conjecture~\ref{conj:cost1} and Proposition \ref{prop:epsilonCost}. Note that the hatched area is not part of the cost region described
by Proposition~\ref{prop:epsilonCost}.

\begin{figure}[t]
\centering
    \includegraphics[width=0.7\textwidth]{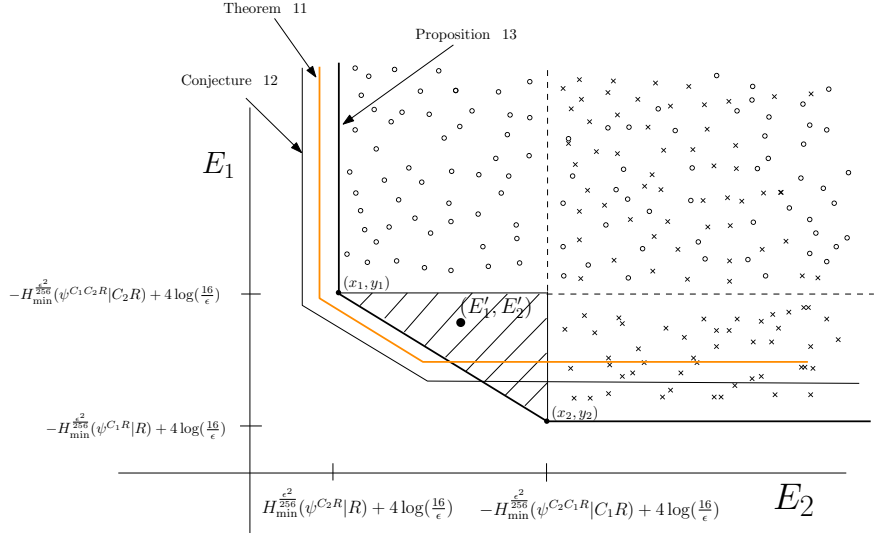}
\caption{Entanglement cost for multiparty merging in the one-shot
regime when $m=2$. The axes correspond to the entanglement cost
$E_1:=\log K_1-\log L_1$ and $E_2:=\log K_2 - \log L_2$. For
$m=2$, we have two permutations of the set $\{1,2\}$, and
according Proposition~\ref{prop:epsilonCost}, two intersecting
regions (circles and crosses) where existence of a $2$-party state
merging of error $\epsilon$ can be shown. }\label{fig:rateregion1}
\end{figure}

\renewcommand{\min}{{\operatorname{min}}}
\renewcommand{\max}{{\operatorname{max}}}
\newcommand{\tr}{{\operatorname{Tr}}}

\section{A Worked Example} \label{sec:example}

The proof of Theorem \ref{thm:cost1} is significantly more complicated than that of Proposition \ref{prop:epsilonCost}. To illustrate the benefits accruing from the additional effort, we will compare the two results' estimates of the costs achievable for merging $C_1$ and $C_2$ to $R$ for states of the form
\begin{equation}
|\psi\rangle^{C_1 C_2 R} =
\frac{1}{\sqrt{H_d}} \sum_{j=1}^d \frac{1}{\sqrt{j}} |j\rangle^{C_1} |\psi_j\rangle^{C_2} |j\rangle^{R},
\end{equation}
where $H_d = \sum_{j=1}^d 1/j$ is the $d$th harmonic number. These are close relatives of the embezzling states introduced in~\cite{Hayden04}, which are useful resources for channel simulation and other tasks~\cite{Harrow,Berta02,Reverse}. They make interesting examples because they have sufficient variation in their Schmidt coefficients that the i.i.d. state merging rates of Theorem \ref{thm:statemerging} are not achievable in the one-shot regime. Nonetheless, our results yield nontrivial one-shot rates that are significantly better than simple teleportation.  We will assume that  $| \langle \psi_i | \psi_j \rangle | \leq \alpha$ for $i \neq j$ and try to and express the rates in terms of $\alpha$. We will assume for convenience that $\alpha > 0$ since, when $\alpha = 0$, the costs are essentially the same as when $\alpha = \Omega(1/d)$.

\subsection*{Protocols from Theorem \ref{thm:cost1}}
Let $(E_1,E_2)$ be a pair of entanglement costs achievable according to Theorem \ref{thm:cost1}. The only constraints on the costs (aside from needing to be the logs of integers) are
\begin{eqnarray}
E_1 &\geq& -H_\min(\psi^{C_1 R} | \psi^R ) + 4 \log ( 1 / \epsilon ) + 12 \label{eqn:ach-e1} \\
E_2 &\geq&  -H_\min(\psi^{C_2 R} | \psi^R ) + 4 \log ( 1 / \epsilon ) + 12 \label{eqn:ach-e2} \\
E_1 + E_2  &\geq&  -H_\min(\psi^{C_1 C_2 R} | \psi^R ) + 4 \log ( 1 / \epsilon ) + 12. \label{eqn:ach-e1+e2}
\end{eqnarray}
To begin, we will find a sufficient condition for the $E_1$ constraint to be satisfied,
so we need to evaluate $H_\min(\psi^{C_1R} | \psi^R )$. Let $\lambda$ be the smallest real number such that $\lambda( I^{C_1} \otimes \psi^R ) - \psi^{C_1 R} \geq 0$. Expanding the operators, that condition is the same as
\begin{equation} \label{eqn:cond-ent-operator}
\sum_{ij} \frac{\lambda - \delta_{ij}}{j} |ij\rangle\langle ij|^{C_1 R}
	-  \sum_i \sum_{ j \neq i} \frac{1}{\sqrt{ij}} |ii\rangle\langle jj|^{C_1 R}
			\langle \psi_j | \psi_i \rangle
	\geq 0,
\end{equation}
where $\delta_{ij}$ is the Kronecker delta function. By the Gershgorin Circle Theorem \cite{Horn,wiki},  the operator will be positive semidefinite if each diagonal entry dominates the sum of the absolute values of the off-diagonal entries in the corresponding row. That condition reduces to
\begin{equation}
\frac{\lambda - 1}{i}
	\geq
	\sum_{j \neq i} \frac{1}{\sqrt{ij}} | \langle \psi_j | \psi_i \rangle |
\end{equation}
holding for all $i$, which is true provided
$\lambda - 1 \geq \alpha \sum_{j=1}^d \sqrt{d/j}$.
But
\begin{equation}
\sum_{j=1}^d \frac{1}{\sqrt{j}}
	\leq
	\int_0^{d} \frac{1}{\sqrt{x}} \, dx
	= 2 \sqrt{d}.
\end{equation}
Therefore, the operator of eq.~(\ref{eqn:cond-ent-operator}) will be positive semidefinite if $\lambda \geq 2 \alpha d + 1$. This in turn implies that
\begin{equation}
-H_\min(\psi^{C_1 R} | \psi^R )
	\leq \log (2\alpha d + 1 )
	\leq \log(\alpha d) + 2.
\end{equation}
The lower bound of eq.~(\ref{eqn:ach-e1}) will therefore be satisfied provided $E_1 \geq \log( \alpha d ) + 4\log(1/\epsilon) + 14$. The interpretation is that if the states $\{ |\psi_j\rangle \}$ are indistinguishable, then $C_1$ holds the whole purification of $R$ and must therefore be responsible for the full cost of merging. As the states $\{ | \psi_j \rangle \}$ become more distinguishable, the purification of $R$ becomes shared between $C_1$ and $C_2$, so the merging cost can be shared. Indeed, if $\alpha = O(1/d)$, then the lower bound on $E_1$ becomes a constant, independent of the size of the input state $|\psi\rangle^{C_1C_2R}$.

Moving on to the $E_2$ constraint, eq.~(\ref{eqn:ach-e2}), a similar but easier calculation shows that $H_\min(\psi^{C_2 R} | \psi^R ) = 0$.
%This translates into the requirement that $E_2 \geq 4\log(1/\epsilon) + 12$. In particular, it is a constant independent of the size of the input state.
For the sum rate $E_1 + E_2$, it is necessary to evaluate $H_\min(\psi^{C_1 C_2 R} | \psi^R )$. Since the rank of $\psi^{C_1 C_2}$ is $d$, this entropy is at least $-\log d$~\cite{Renner01}.
%eq.~(\ref{eqn:ach-e1+e2}) then becomes $E_1 + E_2 \geq \log d + 4 \log(1/\epsilon) + 12$.

So any pair of costs $(E_1,E_2)$ satisfying
\begin{eqnarray}
E_1 &\geq&  \log( \alpha d ) + 4\log(1/\epsilon) + 14 \\
E_2 &\geq& 4\log(1/\epsilon) + 12 \\
E_1 + E_2 &\geq& \log d + 4 \log(1/\epsilon) + 12 \label{eqn:cor12e1e2}
\end{eqnarray}
will be achievable by Theorem \ref{thm:cost1}. The total cost $E_1 + E_2$ must be at least $\log d$ plus terms independent of the size of $|\psi\rangle^{C_1 C_2 R}$ and that cost can be shared between $E_1$ and $E_2$. The lower bound on $E_2$ alone is independent of $d$ and should be regarded as a small ``overhead'' for the protocol. There \emph{is} a minimal $d$-dependent cost for $E_1$, however, which encodes the fact that if $C_2$ does not carry enough of the purification of $R$ by virtue of the nonorthogonality of the $\{|\psi_j\rangle\}$, then more of the burden will fall to $C_1$.

\subsection{Protocols from Proposition \ref{prop:epsilonCost}}
Now let us consider the costs achievable according to Proposition \ref{prop:epsilonCost}. For fixed $\epsilon$, the proposition provides two cost pairs, plus others that are simply degraded versions of those two arising from the wasteful consumption of unnecessary entanglement. Proposition \ref{prop:epsilonCost} does not permit interpolation between the two points, as compared to Theorem \ref{thm:cost1}.
It might be the case, however, that Proposition \ref{prop:epsilonCost}'s freedom to smooth the entropy and vary the operator being conditioned upon could result in those two cost pairs being much better than any of those provided by Theorem \ref{thm:cost1}. On the contrary, for the states of the example, the improvement achieved with the extra freedom is minimal.

Let $(E_1',E_2')$ be a cost pair achievable by Proposition \ref{prop:epsilonCost}. For the purposes of illustration, consider the point with the smallest possible value of $E_2'$. Letting $\delta = \epsilon^2/256$, that point will satisfy
\begin{eqnarray}
E_1' &\geq&  -H_\min^\delta (\psi^{C_1 C_2 R} | C_2 R)
		+ 4 \log\left( 2 / \epsilon \right) + 12  \\
E_2' &\geq&  -H_\min^\delta (\psi^{C_2 R} | R)
		+   4 \log\left( 2 / \epsilon \right) + 12.
\end{eqnarray}
%Using the duality of the min- and max- entropies,
%\begin{eqnarray}
% -H_\min^\delta (\psi^{C_2 R} | R)
% 	&=& H_\max^\delta ( \psi^{C_1 C_2} | C_1 ) \\
%	&=& \min_{\phi \in \mathcal{B}^\delta(\psi^{C_1 C_2})} H_\max( \phi^{C_1 C_2} | C_1 ),
%\end{eqnarray}
%where $B^\delta(\psi^{C_1 C_2})$ is a $\delta$-ball centered at $\psi$.
Since the state $\psi^{C_2 R}$ is separable, the cost $E_2'$ cannot be negative, at least for sufficiently small $\epsilon$, so the key number is the $E_1'$ cost. Before introducing the extra complication of smoothing, consider first $H_\min(\psi^{C_1 C_2 R} | C_1 R )$. By \cite{Renner03}, this is related to the largest overlap that can be achieved with a maximally entangled state on $C_2 / C_1 R$ by acting with a quantum channel on the $C_1 R$ part of $\psi^{C_1 C_2 R}$. This  \emph{maximum singlet fraction} is at least what is achieved by just aligning the Schmidt bases, which is
\begin{eqnarray}
\left| \sum_{j=1}^d \frac{1}{\sqrt{j \cdot H_d}} \cdot \frac{1}{\sqrt{d}} \right|^2
	&=& \frac{1}{d \cdot H_d} \left| \sum_{j=1}^d \frac{1}{\sqrt{j}} \right|^2 \\
	&\geq&  \frac{1}{d \cdot H_d} \left|  \int_{1}^{d} \frac{1}{\sqrt{x}} dx \right|^2 \\
%	&=& \frac{4}{d \cdot H_d} \left(  \sqrt{d}-1 \right )^2 \\
	&=& \frac{4}{H_d} \left( 1 - O\left( \frac{1}{\sqrt{d}} \right) \right) \\
	&\geq&  \frac{5}{\log d},
\end{eqnarray}
where the last line holds for sufficiently large $d$.
Above and in what follows, we use the inequality  $\ln ( d + 1 ) \leq H_d \leq (\ln d)+1$, which was supplemented above by the fact that $4/(\ln d + 1) \geq 5.7 / \log d$ for sufficiently large $d$.
According to Theorem 2 of \cite{Renner03}, the resulting bound on  $H_\min$ is
\begin{equation}
- H_\min( \psi^{C_1 C_2 R} | R C_2 )
%	&\geq& \log \left[ d \left( \frac{5}{ \log d} - O\left( \frac{1}{\sqrt{d}} \right) \right ) \right] \\
	\geq \log d - \log \log d  + 2.
\end{equation}
Therefore, ignoring smoothing, the sum cost for Proposition \ref{prop:epsilonCost} will always satisfy
\begin{equation}
E_1' + E_2'
	\geq   \log d - \log \log d  + 26 + 8 \log\left( \frac{2}{\epsilon} \right),
\end{equation}
for sufficiently large $d$, which has worse constants and even asymptotically only differs from the sum cost (\ref{eqn:cor12e1e2}) for Theorem \ref{thm:cost1} by $O(\log \log d)$.

Now let us introduce some smoothing. By duality of the min- and max- entropies,
\begin{equation}
-H^\delta_\min(\psi^{C_1 C_2 R} | R C_2 ) = H^\delta_\max(\psi^{C_1}).
\end{equation}
Lemma \ref{lem:hmax-smoothing} of Appendix~\ref{app:hmax-smoothing} gives that
\begin{equation}
H^\delta_\max(\psi^{C_1})
	\geq 2 \log \min \left\{ \sum_{j=1}^{k-1} \frac{1}{\sqrt{j \cdot H_d}}: k \mbox{ such that }
		\sum_{j=k+1}^d \frac{1}{j \cdot H_d} \leq \frac{\delta^2}{2} \right\}.
\end{equation}
Getting a lower bound on this expression requires finding large $k$ that nonetheless fail to satisfy the tail condition. That restriction on $k$ is equivalent to $1-H_k/H_d \leq \delta^2/2$, which will not be met by any $k$ small enough to obey
\begin{equation} \label{eqn:kbound}
k \leq (d+1)^{1-\delta^2/2} / e
\end{equation}
for sufficiently large $d$.
Using a similar estimate as for the maximum singlet fraction calculation, we get
\begin{eqnarray}
2 \log \sum_{j=1}^{k-1} \frac{1}{\sqrt{j\cdot H_d}}
	&\geq& \log \left( \frac{1}{\sqrt{H_d}} \int_1^k \frac{1}{\sqrt{x}} \, dx \right)^2 \\
	&\geq& \log \frac{4k}{H_d} \left( 1 - O\left(\frac{1}{\sqrt{k}} \right) \right)  \\
	&\geq& \log k - \log \log d+ \log 5
\end{eqnarray}
for sufficiently large $k$.
Substituting in the largest possible $k$ consistent with eq.~(\ref{eqn:kbound}) and $\delta = \epsilon^2/256$ gives
\begin{equation}
E_1' + E_2' \geq \left( 1 - \frac{\epsilon^4}{512} \right) \log(d+1) -\log \log d + 24 + 8 \log\left(\frac{2}{\epsilon}\right),
\end{equation}
for sufficiently large $d$.
The sum costs achievable using Theorem \ref{thm:cost1} compare favorably with this bound. The additional savings from smoothing are only about $\epsilon^4 \log(d+1) / 512$ ebits, which is insignificant for small $\epsilon$. These tiny savings also come at the expense of being able to interpolate between achievable costs. To be fair, these states were chosen specifically because they are known to maintain their essential character even after smoothing, as was observed in~\cite{Hayden05}. The freedom to smooth is certainly more beneficial for some other classes of states, most notably i.i.d. states. Indeed, since $S(\psi^{C_1 C_2}) = (\log d)/2 + O( \log \log d)$, merging many copies of $|\psi\rangle^{C_1 C_2 R}$ can be done at a rate roughly half the cost required for one-shot merging.

\section{A variant of Merging: Split Transfer} \label{sec:split-transfer}

In the previous sections, we've analyzed and characterized the
entanglement cost for merging the state $\initstate$ to a single
receiver Bob in the asymptotic setting and in the one-shot regime.
Here, we modify our initial setup by introducing a second decoder
$A$ (Alice), who is spatially separated from Bob and also has side
information about the input state. That is, the helpers $C_1,
C_2,\ldots, C_m$ and the two receivers Alice and Bob share a
global state $\psi^{C_1C_2\ldots C_mABR}$ and the objective is
then to redistribute the state $\initstateA$ to Alice, Bob, and
the reference $R$. The motivation for this problem comes from the
multipartite entanglement of assistance problem \cite{SVW,Merge},
where the task is to distill entanglement in the form of EPR pairs
from a $(m+2)$-partite pure state $\inputstate$ shared between two
recipients (Alice and Bob) and $m$ other helpers $C_1,C_2,\ldots,C_m$. If many copies of the input state are
available, the optimal EPR rate was shown in \cite{Merge} to be
equal to
\begin{equation}\label{eq:mincut1}
  E^{\infty}_A(\inputstate) := \min_{\cal T} S(A{\cal T})_{\psi},
\end{equation}
where ${\cal T} \subseteq \{C_1,C_2,\ldots,C_m\}$ is a subset (i.e a bipartite cut) of the helpers. We
denote the complement by $\cT := \{C_1C_2\ldots C_m\} \setminus {\cal T}$. We call $\min_{\cal T}\{S(A{\cal T})_{\psi}\}$ the minimum cut (min-cut) entanglement of the state $\inputstate$.

The proof that the rate given by eq.~(\ref{eq:mincut1}) is
achievable using LOCC operations consists of showing that the
min-cut entanglement of the state $\inputstate$ is preserved, up
to an arbitrarily small variation, after each sender has finished
performing a random measurement on his system. The procedure
described in the proof of \cite{Merge} makes use of a
multiple-blocking strategy. That is, given $n$ copies of the input
state $\inputstate$, the first helper will perform $d=n/m$ random
measurements, each acting on $m$ copies of $\inputstate$ and
generating $J$ possible outcomes. Then, if each measurement can
yield outcomes $j_1,j_2,\ldots,j_d$, we need to group together the
residual states corresponding to outcome $1$, then group the ones
corresponding to outcome $2$, etc... When this is done, the next
helper will perform random measurements for each of these groups in
the same way the first sender proceeded. That is, for each group,
you need to divide into blocks, and so on. Needless to say, this approach fails in the one-shot setting.

It was conjectured in \cite{Merge} that these layers of blocking
could be removed by letting all the helpers perform simultaneous
measurements on their respective typical subspaces. Such a
strategy would still produce states which preserve the min-cut
entanglement, thereby providing a way to prove
eq.~(\ref{eq:mincut1}) without the need for a recursive argument.
We will show in the remainder of this section that for a cut
${\cal T}_{\min}$ which minimizes $S(A\cK)_{\psi}$, there exists
an LOCC protocol acting on the state $\initstate$ which will send
${\cal T}_{\min}$ to Alice and its complement to Bob. The protocol
will consist of two parts. First, all the helpers will perform
measurements on their typical subspaces and broadcast their
outcomes to both decoders. Then, Alice will use the classical
information coming from the helpers which are part of the cut
${\cal T}_{\min}$ and apply an isometry $U$, while Bob will apply
an isometry $V$ depending on the outcomes of the helpers belonging
to $\cT_{\min}$. This will redistribute the initial state to
Alice, Bob, and the reference $R$. Standard distillation protocols
\cite{Concentrate,Bennett} on the recovered state will yield
EPR pairs at a rate given by eq.~($\ref{eq:mincut1}$).

\begin{definition}
Let $\psi^{{\cal T}{\ov{\cal T}}ABR}$ be an $(m+2)$-partite state,
where ${\cal T}$ and $\cT$ are a partition of the helpers $C_1,
C_2,\ldots, C_m$. Furthermore, assume that the helpers and the
decoders share maximally entangled states $\Phi^K:=\bigotimes_{i
\in \cal T} \Phi^{K_i}$ and $\Gamma^M:=\bigotimes_{i \in \cT}
\Gamma^{M_i}$ on the tensor product spaces $\cK^0A^0_{\cK}$ and
$\cTo B^0_{\cT}$.

We call the LOCC operation ${\cal M}: {\cal T}\cK^0 \ov{\cal
T}\cTo \otimes AA^0_{{\cal T}} \otimes BB^0_{\ov{{\cal T}}}
\rightarrow {\cal T}^1A^1_{{\cal T}}AA_{{\cal T}} \otimes
{\ov{\cal T}}^1B^1_{\ov{{\cal T}}}BB_{\ov{{\cal T}}} $ a
\emph{split transfer} of the state $\inputGroupstate$ with error
$\epsilon$ and associated entanglement costs
$\overrightarrow{E_{\cK}}(\psi):=\bigoplus_{i \in \cK}(\log K_i
-\log L_i)$ and $\overrightarrow{E_{\cKbar}}(\psi):=\bigoplus_{i
\in \cKbar}(\log M_i - \log N_i)$ if
 \begin{equation}\label{eq:splittransfer}
 \bigg \| (\mathrm{id}_R \otimes {\cal M})(\inputGroupstate \otimes \Phi^K  \otimes \Gamma^M) -  \psi_{AA_{{\cal T}}BB_{\ov{\cal T}}R} \otimes \Phi^L \otimes \Gamma^{N} \bigg \|_1 \leq \epsilon,
\end{equation}
where $\Phi^L:=\bigotimes_{i \in \cK}\Phi^{L_i},\Gamma^N:=\bigotimes_{j \in \cKbar}\Gamma^{N_j}$, with the states $\Phi^{L_i}$ and $\Gamma^{N_j}$ being maximally entangled states of Schmidt ranks $L_i$ and $N_i$ on $C^1_iA^1_i$ and $C_j^1B^1_j$ respectively. Also, the systems $A_{\cal T}$ and $B_{\cT}$ are ancillary systems of the same size as ${\cal T}$ and $\cKbar$ and are held by Alice and Bob respectively. For the state $\Psi:=(\inputGroupstate)^{\otimes n}$, the entanglement rates $\overrightarrow{R_{\cK}}(\psi)$ and $\overrightarrow{R_{\cKbar}}(\psi)$ are defined as $\frac{1}{n}\overrightarrow{E_{\cKbar}}(\Psi)$ and
$\frac{1}{n}\overrightarrow{E_{\cKbar}}(\Psi)$.
\end{definition}

In the above definition, we have denoted by $\bigoplus_{i \in \cK} (\log K_i - \log L_i)$ a vector of length $|\cK|$ whose components are given by $\log K_i -\log L_i$ for $i \in \cK$ in the lexicographical order.

The rate region where a split-transfer can be accomplished by LOCC
can be defined in a manner analogous to definition
\ref{def:rateregion}. We omit the details here, but whenever we
will say that a rate is achievable for a split-transfer of the
state $\inputGroupstate$, it will mean that it is contained in the
rate region.

Now, we'd like to specify conditions, as in Proposition
\ref{prop:mergeCond}, that the initial state should satisfy in
order to allow the group ${\cal T}$ (resp. ${\ov{\cal T}}$) to
transfer their share of the state to Alice (resp. Bob). For a pure
state $\psi^{\cK\cKbar ABR}$, suppose all the helpers perform
incomplete measurements (as in Section~\ref{sec:merging-many}) on their respective
shares of the state. For measurement outcomes $J :=
(j_1,j_2,\ldots,j_m)$, define the state
\begin{equation}
 \begin{split}
 \ket{\psi_J^{\cK^1\cTu ABR}} &:= \frac{1}{\sqrt{p_J}}(P^1_{j_1} \otimes P^2_{j_2} \otimes \ldots \otimes P^m_{j_m} \otimes I^{ABR}) \ket{\psi^{\cK\cKbar ABR}} \\
   & =: \frac{1}{\sqrt{p_J}}(P^{\cK}_{j_{\cK}} \otimes P^{\cKbar}_{j_{\cKbar}} \otimes I^{ABR}) \ket{\psi^{\cK\cKbar ABR}}, \\
 \end{split}
\end{equation}
where $p_J$ is the probability of getting outcome $J$. In the above definition, $j_{\cK}$ is a vector of length $|\cK|$ whose components correspond to outcomes of measurements performed by the helpers belonging to the cut $\cK$. The vector $j_{\cKbar}$ is defined similarly. Finally, the Kraus operators $P^{\cK}_{j_{\cK}} = \bigotimes_{i \in \cK} P^i_{j_i}$ and $P^{\cKbar}_{j_{\cKbar}} = \bigotimes_{i \in \cKbar} P^i_{j_i}$ map the spaces $\cK$ and $\cKbar$ to the subspaces $\cK^1$ and $\cTu$ respectively.

Define another state $\ket{\varphi_{j_{\cK}}^{\cK^1\cKbar ABR}} :=
\frac{1}{\sqrt{p_{j_{\cK}}}}(P^{\cK}_{j_{\cK}} \otimes I^{\cKbar
ABR})\ket{\psi^{\cK\cKbar ABR}}$, where $p_{j_{\cK}}$ is the
probability of getting the outcome $j_{\cK}$, and suppose that we
have
\begin{equation}
 \varphi_{j_{\cK}}^{\cK^1\cKbar BR} = \tau_L^{\cK^1} \otimes \psi^{\cKbar BR},
\end{equation}
where $\tau_L^{\cK^1} = \bigotimes_{i \in \cK} \tau^{C^1_i}$ is
the maximally mixed state of dimension $L$ on the system $\cK^1$.
From the Schmidt decomposition, we know there exists an isometry
$U^A_{j_{\cK}}: A \rightarrow A^1_{\cK} A_{\cK} A$ which Alice can
perform such that
\begin{equation}
(I^{\cK^1 \cKbar BR} \otimes U^A_{j_{\cK}}) \ket{\varphi_{j_{\cK}}^{\cK^1\cKbar ABR}} = \ket{\Phi^L} \otimes \ket{\psi^{ A_{\cK} \cKbar ABR}},
\end{equation}
where the state $\ket{\psi^{A_{\cK} \cKbar ABR}}$ is the same as
the original state $\ket{\psi^{\cK \cKbar ABR}}$ with the
ancillary system $A_{\cK}$ substituted for $\cK$. The state
$\ket{\Phi^L}$ is a maximally entangled state on $\cK^1A^1_{\cK}$.

Finally, define the state $\ket{\upsilon_{j_{\cKbar}}^{A_{\cK}
\cTu ABR}}
:=\frac{1}{\sqrt{p_{j_{\cKbar}}}}(P^{\cKbar}_{j_{\cKbar}} \otimes
I_{A_{\cK}ABR})\ket{\psi^{A_{\cK}\cKbar A BR}}$ and suppose again
that we have
\begin{equation}
 \upsilon_{j_{\cKbar}}^{A_{\cK} \cTu AR} = \tau_N^{\cTu} \otimes \psi^{A_{\cK} AR},
\end{equation}
where $\tau_N^{\cTu} = \bigotimes_{i \in \cKbar} \tau^{C^1_i}$ is
the maximally mixed state of dimension $N$ on the system $\cTu$.
Applying the Schmidt decomposition once more, Bob can perform an
isometry $V^B_{j_{\cKbar}}: B \rightarrow B^1_{\cKbar} B_{\cKbar}
B$ such that
\begin{equation}
(I^{A_{\cK} \cTu AR} \otimes V^B_{j_{\cKbar}}) \ket{\upsilon_{j_{\cKbar}}^{A_{\cK}\cTu ABR}} = \ket{\Gamma^N} \otimes \ket{\psi^{A_{\cK} B_{\cKbar} ABR}},
\end{equation}
where the state $\ket{\psi^{A_{\cK} B_{\cKbar} ABR}}$ is the same
as the original state $\inputGroupstate$ with the ancillary
systems $A_{\cK}$ and $B_{\cKbar}$ substituted for $\cK$ and
$\cKbar$.

If we apply the isometries $U^A_{j_{\cK}}$ and $V^B_{j_{\cKbar}}$ to the outcome state $\ket{\psi_J^{\cK^1 \cTu ABR}}$, the resulting state is given by
\begin{equation}
 \begin{split}
   &(I^{\cK^1 \cTu R} \otimes U^A_{j_{\cK}} \otimes V^B_{j_{\cKbar}}) \ket{\psi_J^{\cK^1 \cTu ABR}}\\
   &=\frac{1}{\sqrt{p_J}}(I^{\cK^1 \cTu R} \otimes U^A_{j_{\cK}} \otimes V^B_{j_{\cKbar}})( P^{\cKbar}_{j_{\cKbar}} \otimes I^{\cK^1 ABR})(P^{\cK}_{j_{\cK}} \otimes I^{\cKbar ABR})\ket{\psi^{\cK \cKbar ABR}} \\
   &=\frac{1}{\sqrt{p_J}}(I^{\cK^1 A^1_{\cK}A_{\cK} \cTu AR} \otimes V^B_{j_{\cKbar}})(P^{\cKbar}_{j_{\cKbar}} \otimes I^{\cK^1 A^1_{\cK}A_{\cK}ABR} )(I^{\cK^1 BR} \otimes U^A_{j_{\cK}})(P^{\cK}_{j_{\cK}} \otimes I^{\cKbar ABR})\ket{\psi^{\cK \cKbar ABR}} \\
   &=\frac{1}{\sqrt{p_J}}(I^{\cK^1 A^1_{\cK}A_{\cK} \cTu AR} \otimes V^B_{j_{\cKbar}})(P^{\cKbar}_{j_{\cKbar}} \otimes I^{\cK^1 A^1_{\cK}A_{\cK}ABR} )(I^{\cK^1 BR} \otimes U^A_{j_{\cK}}) \sqrt{p_{j_{\cK}}}\ket{\varphi_{j_{\cK}}^{\cK^1\cKbar ABR}}\\
   &=\sqrt{\frac{p_{j_{\cK}}}{p_J}}(I^{\cK^1 A^1_{\cK}A_{\cK} \cTu AR} \otimes V^B_{j_{\cKbar}})(P^{\cKbar}_{j_{\cKbar}} \otimes I^{\cK^1 A^1_{\cK}A_{\cK}ABR}) \ket{\Phi^L} \otimes \ket{\psi^{A A_{\cK} \cKbar BR}} \\
   &=\sqrt{\frac{p_{j_{\cK}}}{p_J}}(I^{\cK^1 A^1_{\cK}A_{\cK} \cTu AR} \otimes V^B_{j_{\cKbar}}) \ket{\Phi^L} \otimes \sqrt{p_{j_{\cKbar}}}\ket{\upsilon_{j_{\cKbar}}^{A_{\cK}\cTu ABR}} \\
   &=\sqrt{\frac{p_{j_{\cK}}p_{j_{\cKbar}}}{p_J}}\ket{\Phi^L} \otimes \ket{\Gamma^N} \otimes \ket{\psi^{A A_{\cK} B_{\cKbar} BR}}. \\
 \end{split}
\end{equation}
Since the states $(I^{\cK^1 \cTu R} \otimes U^A_{j_{\cK}} \otimes
V^B_{j_{\cKbar}}) \ket{\psi_J^{\cK^1 \cTu ABR}}$ and $\ket{\Phi^L}
\otimes \ket{\Gamma^N} \otimes \ket{\psi^{A A_{\cK} B_{\cKbar}
BR}}$ are both normalized, we must have $p_J =
p_{j_{\cK}}p_{j_{\cKbar}}$. Hence, in this ideal case, we can
achieve a split transfer of the state $\initstate$ by letting all
the helpers measure their share simultaneously. The decoding by
Alice and Bob will follow once they receive the measurement
outcomes.

\begin{Proposition}[Conditions for a Split-Transfer] \label{prop:mergeConds}
Denote the state shared between $m$ helpers and two receivers (Alice and Bob) by $\inputGroupstate$, with purifying system $R$. Suppose all the helpers perform incomplete measurements on their share of the state $\inputGroupstate$ as in the previous paragraphs, yielding a state $\ket{\psi_J^{\cK^1\cTu ABR}}:=\frac{1}{\sqrt{p_J}}(P^{\cK}_{j_{\cK}} \otimes P^{\cKbar}_{j_{\cKbar}} \otimes I^{ABR})(\ket{\psi^{\cK\cKbar ABR}})$ for an outcome $J:=(j_1,j_2,\ldots,j_m)$ with probability $p_J$, where the Kraus operators $P^{\cK}_{j_{\cK}}$ and $P^{\cKbar}_{j_{\cKbar}}$ map the spaces $\cK$ and $\cKbar$ to the subspaces $\cK^1$ and $\cTu$.

If, for the quantum errors $Q^1_{\cal I}(\inputGroupstate)$ and $Q^2_{\cal I}(\inputGroupstate)$, we have
 \begin{equation}
    \begin{split}
        Q^1_{\cal I}(\inputGroupstate) &:= \sum_{j_{\cK}} p_{j_{\cK}}\|\varphi_{j_{\cK}}^{\cK^1 \cKbar BR} - \tau^{\cK^1}_L \otimes \psi^{\cKbar BR}\|_1 \leq \epsilon \\
        Q^2_{\cal I}(\inputGroupstate) &:= \sum_{j_{\cKbar}} p_{j_{\cKbar}}\|\upsilon_{j_{\cKbar}}^{A_{\cK} \cTu AR} - \tau^{\cTu}_N \otimes \psi^{A_{\cK} AR}\|_1 \leq \epsilon', \\
    \end{split}
 \end{equation} then there exists a split-transfer of the state $\inputGroupstate$ with error $2\sqrt{\epsilon}+2\sqrt{\epsilon'}$ and entanglement costs $\overrightarrow{E_{\cK}}=\bigoplus_{i \in \cK}(-\log L_i)$ and $\overrightarrow{E_{\cKbar}}=\bigoplus_{i \in \cKbar}(-\log N_i)$. The states $\varphi_{j_{\cK}}^{\cK^1 \cKbar BR}$ and $\upsilon_{j_{\cKbar}}^{A_{\cK} \cTu AR}$ are reduced density operators for the states $\ket{\varphi_{j_{\cK}}^{\cK^1 \cKbar ABR}}:=\frac{1}{\sqrt{p_{j_{\cK}}}}(P^{\cK}_{j_{\cK}} \otimes I^{\cKbar ABR})\ket{\psi^{\cK\cKbar ABR}}$ and $\ket{\upsilon_{j_{\cKbar}}^{A_{\cK} \cTu ABR}}:=\frac{1}{\sqrt{p_{j_{\cKbar}}}}(P^{\cKbar}_{j_{\cKbar}} \otimes I^{A_{\cK}ABR})\ket{\psi^{A_{\cK}\cKbar A BR}}$.
\end{Proposition}

\begin{proof}
Since the quantum errors $Q^1_{\cal I}$ and $Q^2_{\cal I}$ are
bounded from above by $\epsilon$ and $\epsilon'$ respectively, Proposition \ref{prop:mergeCond}
can be applied, which tells us of the existence of isometries
$U^A_{j_{\cK}}$ and $V^B_{j_{\cKbar}}$ such that
 \begin{equation}\label{eq:qerror}
   \begin{split}
     \bigg \| \sum_{j_{\cK}} p_{j_{\cK}} (I^{\cKbar BR} \otimes U^A_{j_{\cK}}) \braket{\varphi^{\cK^1 \cKbar ABR}_{j_{\cK}}}(I^{\cKbar BR} \otimes U^A_{j_{\cK}})^{\dag} - \psi^{A_{\cK} \cKbar ABR} \otimes \Phi^L \bigg \|_1 &\leq 2\sqrt{\epsilon} \\
     \bigg \| \sum_{j_{\cKbar}} p_{j_{\cKbar}} (I^{A_{\cK} AR} \otimes V^B_{j_{\cKbar}}) \braket{\upsilon^{A_{\cK} \cTu ABR}_{j_{\cKbar}}}(I^{A_{\cK} AR} \otimes V^B_{j_{\cKbar}})^{\dag} - \psi^{A_{\cK} B_{\cKbar} ABR} \otimes \Gamma^N \bigg \|_1 &\leq 2\sqrt{\epsilon'}.\\
   \end{split}
 \end{equation}
 If we apply the isometries $U^A_{j_{\cK}}$ and $V^B_{j_{\cKbar}}$ to the state $\ket{\psi^{\cK^1 \cTu ABR}_J}$ after obtaining outcome $J$, the output state $\rho$ of the protocol will be of the form
% \rho^{\cK^1A^1_{\cK} \cTu B^1_{\cKbar}A_{\cK}B_{\cKbar} ABRX}
 \begin{equation}
  \begin{split}
    &\rho := \sum_J p_J \bigg ( (I^{\cK^1 \cTu R} \otimes \UA \otimes \VB) \psi^{\cK^1 \cTu ABR}_J (I^{\cK^1 \cTu R} \otimes \UA \otimes \VB)^{\dag} \bigg )\\
     &= \sum_J p_J \bigg ( \frac{p_{j_{\cK}}}{p_J}(I \otimes \VB)(P^{\cKbar}_{\jKbar}\otimes I)(I \otimes \UA)\braket{\varphi_{j_{\cK}}^{\cK^1 \cKbar ABR}}(I \otimes \UA)^{\dag}(P^{\cKbar}_{\jKbar} \otimes I)^{\dag}(I \otimes \VB)^{\dag} \bigg )\\
     %&=\sum_{\jKbar} (I \otimes \VB)(P^{\cKbar}_{\jKbar}\otimes I) \bigg [ \sum_{\jK} (I \otimes \UA)(P^{\cK}_{\jK} \otimes I)\psi^{\cK \cKbar ABR} (P^{\cK}_{\jK} \otimes I)^{\dag}(I \otimes \UA)^{\dag} \otimes \braket{\jK}^{X_1} \bigg] (P^{\cKbar}_{\jKbar} \otimes I)^{\dag}(I \otimes \VB)^{\dag} \otimes \braket{\jKbar}^{X_2} \\
     &=\sum_{\jKbar} (I \otimes \VB)(P^{\cKbar}_{\jKbar}\otimes I) \zeta (P^{\cKbar}_{\jKbar} \otimes I)^{\dag}(I \otimes \VB)^{\dag}  \\
     &=: {\cal M}(\zeta)\\
  \end{split}
 \end{equation}
where $\zeta := \sum_{\jK} p_{j_{\cK}} (I \otimes \UA)\braket{\varphi_{j_{\cK}}^{\cK^1 \cKbar ABR}}(I \otimes \UA)^{\dag}$. It can be seen as the output state we would get if only the helpers in $\cK$ wanted to transfer their share of the state to the decoder $A$. The map ${\cal M}$, as defined above, corresponds to an LOCC quantum operation acting on $\zeta$ which consists of measurements by the helpers in ${\cKbar}$ followed by an isometry on $B$. %The system $X:=X_1X_2$ is an ancillary system holding the measurement outcomes and $q_{\jKbar}$ is the probability of getting outcome $\jKbar$ if the parties in $\cKbar$ perform incomplete measurements on the state $\zeta$.
Note that we have remove some of the superscript notation for the sake of clarity.

We would like to bound the trace distance between the output state
$\rho$ and the state $\psi^{A_{\cK}B_{\cKbar}ABR} \otimes \Phi^L
\otimes \Gamma^N$. To achieve this, we introduce the following
intermediate state
 \begin{equation}
  \begin{split}
   \sigma &:= \sum_{\jKbar}(I \otimes \VB)(P^{\cKbar}_{\jKbar} \otimes I)(\psi^{A_{\cK}\cKbar ABR} \otimes \Phi^L)(P^{\cKbar}_{\jKbar} \otimes I)^{\dag}(I \otimes \VB)^{\dag} \\
   &= {\cal M}(\psi^{A_{\cK}\cKbar ABR} \otimes \Phi^L) \\
  \end{split}
 \end{equation}
and apply the triangle inequality
\begin{equation}
 \begin{split}
   \bigg \| \rho - &\psi^{A_{\cK}B_{\cKbar}ABR} \otimes \Phi^L \otimes \Gamma^N \bigg \|_1 \leq \bigg \|\rho - \sigma \bigg \|_1 +\bigg \| \sigma - \psi^{A_{\cK}B_{\cKbar}ABR} \otimes \Phi^L \otimes \Gamma^N\bigg\|_1. \\
 \end{split}
\end{equation}
The trace norm $\bigg \| \sigma - \psi^{A_{\cK}B_{\cKbar}ABR}
\otimes \Phi^L \otimes \Gamma^N \bigg \|_1$ is equal to the trace
norm appearing in the second line of eq.~(\ref{eq:qerror}), and so
is bounded from above by $2\sqrt{\epsilon'}$. To bound $\bigg
\|\rho - \sigma \bigg \|_1$, we have
 \begin{equation}
  \begin{split}
  \bigg \| \rho - \sigma \bigg \|_1 &= \bigg \| {\cal M}(\zeta) - {\cal M}(\psi^{A_{\cK}\cKbar ABR} \otimes \Phi^L) \bigg \|_1 \\
    &\leq  \bigg \| \zeta - \psi^{A_{\cK}\cKbar ABR} \otimes \Phi^L \bigg \|_1 \\
    &\leq 2\sqrt{\epsilon}.
  \end{split}
 \end{equation}
 The first inequality holds since the trace distance is non-increasing under quantum operations, and the second inequality is just the first part of eq.~(\ref{eq:qerror}).
Thus, we have a split-transfer of the state $\initstateA$ with error $2\sqrt{\epsilon}+2\sqrt{\epsilon'}$.
\end{proof}

With this result in hand, a one-shot split-transfer protocol of the state $\inputGroupstate$ where all the helpers perform simultaneous random measurements on their share can be obtained by two independent applications of Proposition \ref{prop:isometry} followed by an application of Proposition \ref{prop:mergeConds}. We state the result here.

\begin{Proposition}[One-Shot Split-Transfer] \label{prop:rdnSplit}
Let $\inputGroupstate$ be an $(m+2)$-partite pure state, with purifying system $R$ and local dimensions $d_A,d_{B},d_{R}$. Furthermore, let $d_{\cK}:=\Pi_{i \in \cK}d_{C_i}$ and $d_{\cKbar}:=\Pi_{i \in \cKbar}d_{C_i}$ be the dimensions of the systems $\cK$ and $\cKbar$. Finally, allow the helpers to share additional maximally entangled states $\Phi^K$ and $\Gamma^M$ with the decoders.

For each party $C_i$ in the cut $\cK$, there exists an instrument ${\cal I}_i=\{\calE^i_j\}_{j=0}^{F_i}$ consisting of $F_i
= \lfloor \frac{d_{C_i}K_i}{L_i} \rfloor$ partial isometries of rank $L_i$ and one of rank $L_i'= d_{C_i}K_i - F_i
L_i < L_i$ such that the overall quantum error $Q^1_{{\cal I}}(\inputGroupstate \otimes \Phi^K)$ is bounded by
\begin{equation} \label{eq:upperbound1}
  \begin{split}
  Q^1_{\cal I}(\inputGroupstate \otimes \Phi^K) &\leq 2 \sum_{\substack{{\cal S} \subseteq \cK \\ {\cal S} \neq \emptyset}} \prod_{i \in {\cal S}}\frac{L_i}{d_{C_i}K_i} + 2\sqrt{d_{\cKbar}d_Bd_R \sum_{\substack{{\cal S} \subseteq \cK \\ {\cal S} \neq \emptyset}}  \prod_{i \in {\cal S}} \frac{L_i}{K_i} \Tr \bigg [ (\psi^{{\cal S}\cKbar BR})^2 \bigg ]} =: \Delta^1_{\cal I}. \\
  \end{split}
\end{equation}

Similarly, for each helper $C_i$ in the cut $\cKbar$, there exists an instrument ${\cal I}_i=\{\calE^i_j\}_{j=0}^{G_i}$ consisting of $G_i
= \lfloor \frac{d_{C_i}M_i}{N_i} \rfloor$ partial isometries of rank $N_i$ and one of rank $N_i'= d_{C_i}M_i - G_i
N_i < N_i$ such that the overall quantum error $Q^2_{{\cal I}}(\inputGroupstate \otimes \Gamma^M)$ is bounded by
\begin{equation} \label{eq:upperbound2}
  \begin{split}
  Q^2_{\cal I}(\inputGroupstate \otimes \Gamma^M) &\leq 2 \sum_{\substack{{\cal S} \subseteq \cKbar \\ {\cal S} \neq \emptyset}} \prod_{i \in {\cal S}}\frac{N_i}{d_{C_i}M_i} + 2\sqrt{d_{A_{\cK}}d_Ad_R \sum_{\substack{{\cal S} \subseteq \cKbar \\ {\cal S} \neq \emptyset}}  \prod_{i \in {\cal S}} \frac{N_i}{M_i} \Tr \bigg [ (\psi^{A_{\cK}{\cal S}AR})^2 \bigg ]} =: \Delta^2_{\cal I}. \\
  \end{split}
\end{equation}
Then, there exists a split-transfer of the state $\inputGroupstate$
with error $ 2\sqrt{\Delta^1_{\cal I}}+2\sqrt{\Delta^2_{\cal I}}$.
The left hand sides of eqs.~(\ref{eq:upperbound1}) and (\ref{eq:upperbound2}) are bounded
from above on average by their right hand sides if we perform random
measurements on all the helpers according to the Haar measure.
\end{Proposition}
\begin{proof}
The bound on the quantum errors $Q_{\cal I}^1$ and $Q_{\cal I}^2$
given by eqs.~(\ref{eq:upperbound1}) and (\ref{eq:upperbound2})
can be obtained by two independent applications of Proposition
\ref{prop:isometry} to our setting. We leave the details to the
reader. The existence of a split-transfer with error
$2\sqrt{\Delta^1_{\cal I}}+2\sqrt{\Delta^2_{\cal I}}$ will then
follow from Proposition \ref{prop:mergeConds}. Note here that
since the helpers have additional entanglement at their disposal,
the partial isometries $P^{\cK}_{\jK}$ and $P^{\cKbar}_{\jKbar}$
in Proposition \ref{prop:mergeConds} are replaced by
$P^{\cK\cK^0}_{\jK}$ and $P^{\cKbar\cTo}_{\jKbar}$. These will act
on the spaces $\cK\cK^0$ and $\cKbar\cTo$ respectively, with
output spaces corresponding to $\cK^1$ and $\cTu$.
\end{proof}

Similarly, for the i. i. d. version, we can treat
each quantum error independently and follow a  line of reasoning similar to that in Section~\ref{sec:one-shot}. We arrive at a
variation on Theorem \ref{thm:statemerging}:

\begin{Theorem}[$m$-Party Split-Transfer] \label{thm:splittransfer}
Let $\inputGroupstate$ be a purified state which is shared between $m$ helpers and two receivers (Alice and Bob), with purifying system $R$. For all non-empty subsets ${\cal X} \subseteq {\cal T}$ and ${\cal Y} \subseteq {\cKbar}$, define ${\cal X}$ and ${\cal Y}$ as the tensor products $\bigotimes_{i \in {\cal X}}C_i$ and $\bigotimes_{i \in {\cal Y}}C_i$. Then, the rates $\overrightarrow{R_{\cK}}(\psi):=\bigoplus_{i \in \cK}(R_i)$ and $\overrightarrow{R_{\cKbar}}(\psi):=\bigoplus_{i \in \cKbar}(R_i)$ are achievable for a split-transfer of $\inputGroupstate$ iff the following inequalities
\begin{align}
   \label{eq:splitcond1}\sum_{i \in {\cal X}} R_i &\geq S({\cal X}|{\overline{\cal X}}A)_{\psi}\\
   \label{eq:splitcond2}\sum_{i \in {\cal Y}} R_i &\geq S({\cal Y}|{\overline{\cal Y}}B)_{\psi}
\end{align}
hold for all non-empty subsets ${\cal X} \subseteq {\cK}$ and ${\cal Y} \subseteq {\cKbar}$. The systems ${\overline{\cal X}}$
and ${\overline{\cal Y}}$ are defined as the complements of ${\cal X}$ and ${\cal Y}$ with respect to the systems ${\cK}$ and ${\cKbar}$ respectively.\end{Theorem}

\begin{proof}
To prove achievability, we can proceed exactly as in the proof of Theorem \ref{thm:statemerging}. That is, we Schumacher compress the state $(\inputGroupstate)^{\otimes n}$, and then perform random measurements on the helpers with the following bounds on the ranks of the projectors and of the pre-shared entanglement:
    \begin{align}
         \label{eq:bound1} 0 \leq \prod_{i \in {\cal X}} \frac{L_i}{K_i} &\leq 2^{n(S({\cal X}\cKbar B R)_{\psi} - S(\cKbar B R)_{\psi} - 3\delta|{\cal X}|)} \\
         \label{eq:bound2} 0 \leq \prod_{i \in {\cal Y}} \frac{M_i}{N_i} &\leq 2^{n(S({\cal Y}\cK A R)_{\psi} - S(\cK A R)_{\psi} - 3\delta|{\cal Y}|)}
    \end{align}
for all non empty subsets ${\cal X} \subseteq \cK$ and ${\cal Y} \subseteq \cKbar$. The bounds on the quantum errors $Q^1_{\cal I}$ and $Q^2_{\cal I}$ given in Proposition \ref{prop:rdnSplit} can then be made arbitrarily small. That is, we will have $Q^1_{\cal I}$ and $Q^2_{\cal I}$ bounded from above by $O(2^{-n\delta/2})$ for some typicality parameter $\delta$. By applying Proposition \ref{prop:mergeConds}, we get a split-transfer of the state $\inputGroupstate$ with error $O(2^{-n\delta/4})$ and entanglement costs $\overrightarrow{E_{\cK}}(\psi)=\bigoplus_{i \in \cK}(\log K_i - \log L_i)$ and $\overrightarrow{E_{\cKbar}}(\psi)=\bigoplus_{i \in \cKbar}(\log M_i - \log N_i)$. These will satisfy
  \begin{equation}
    \begin{split}
     \sum_{i \in {\cal X}}\frac{1}{n}(\log K_i - \log L_i) &\geq S({\cal X}|{\overline{\cal X}}A) + 3\delta|{\cal X}| \\
     \sum_{i \in {\cal Y}}\frac{1}{n}(\log M_i - \log N_i) &\geq S({\cal Y}|{\overline{\cal Y}}B) + 3\delta|{\cal Y}| \\
    \end{split}
  \end{equation}
for all non-empty subsets ${\cal X} \subseteq \cK$ and ${\cal Y}
\subseteq \cKbar$. An application of the gentle measurement lemma
and the triangle inequality then tell us that we can apply the
same protocol on the state $(\inputGroupstate)^{\otimes n}$ and
obtain a split-transfer with error
$O(2^{-n\delta/4})+O(2^{-cn\delta^2/2})$. Since this error goes to
zero as $n$ tends to infinity and $\delta$ was arbitrarily chosen,
we get back the direct part of the statement of the theorem.

To get the converse, we can consider any cut ${\cal X}$ of the
helpers in $\cK$ and look at the preservation of the entanglement
across the cut $A\cX$ vs $\X B\cKbar R$.  We assume, for technical reasons, that $L_i \leq 2^{O(n)}$ for all $i \in \cK$. The initial entropy of entanglement across the cut $A\cX$ vs $\X B\cKbar R$ is
 \begin{equation}\label{eq:Ein}
    E_{in} := n S(A\cX)_{\psi} + \sum_{i \in \X} \log K_i.
 \end{equation}
At the end of any LOCC operation on the state $(\inputGroupstate)^{\otimes n}$, the output state can be seen as an ensemble $\{q_k,\psi^k_{\cK^1A^1_{\cK}A^nA^n_{\cK}\cKbar^1B^1_{\cKbar}B^nB^n_{\cKbar}R^n}\}$ of pure states. Using monotonicity of the entropy of entanglement under LOCC, we have
\begin{equation}
n S(A\cX)_{\psi} + \sum_{i \in \X} \log K_i \geq \sum_k q_k S(\cX^1A^1_{\cK} A^n A^n_{\cK})_{\psi^k},
\end{equation}
where $\cX^1:=\bigotimes_{i \in \cX} C^1_i$. For any LOCC operation performing a split-transfer of the state $(\inputGroupstate)^{\otimes n}$ with error $\epsilon$, we have
\begin{equation}
 \sum_k q_k F^2(\psi^k_{\cK^1A^1_{\cK}A^nA^n_{\cK}\cKbar^1B^1_{\cKbar}B^nB^n_{\cKbar}R^n}, \psi^{\otimes n}_{AA_{{\cal T}}BB_{\ov{\cal T}}R} \otimes \Phi^L \otimes \Gamma^{N}) \geq (1-\epsilon/2)^2.
\end{equation}
This follows from the definition of a split-transfer (eq.~(\ref{eq:splittransfer})) and the fact that $F^2$ is linear when one argument is pure. Using Lemma \ref{Lemma:relation}, we can rewrite this as
\begin{equation}
 \sum_k q_k \bigg \|\psi^k_{\cK^1A^1_{\cK}A^nA^n_{\cK}\cKbar^1B^1_{\cKbar}B^nB^n_{\cKbar}R^n} - \psi^{\otimes n}_{AA_{{\cal T}}BB_{\ov{\cal T}}R} \otimes \Phi^L \otimes \Gamma^{N} \bigg \| \leq 2\sqrt{\epsilon(1-\epsilon/2)}.
\end{equation}
By monotonicity of the trace norm under partial tracing, we get
\begin{equation}
 \sum_k q_k \bigg \|\psi^k_{\cX^1A^1_{\cK}A^nA^n_{\cK}} - \psi^{\otimes n}_{AA_{{\cal T}}} \otimes \tau_{A_{\X}^1} \otimes \bigotimes_{i \in \cX} \Phi^{L_i} \bigg \| \leq 2\sqrt{\epsilon(1-\epsilon/2)}.
\end{equation}
Using the Fannes inequality (Lemma \ref{lem:Fannes}) and the concavity of the $\eta$-function, we have
\begin{equation}
 \begin{split}
 \sum_k q_k \bigg |S(\cX^1A^1_{\cK}A^nA^n_{\cK})_{\psi^k} - \sum_{i \in \X} \log L_i - n S(A\X\cX)_{\psi} \bigg | & \leq (2\sum_{i \in \cK}\log L_i + n\log d_A + n\log d_{A_{\cK}}) \eta(2\sqrt{\epsilon(1-\epsilon/2)})\\  &\leq O(n)\eta(2\sqrt{\epsilon(1-\epsilon/2)}).\\
 \end{split}
\end{equation}
Finally, using eq.~(\ref{eq:Ein}), we have
\begin{equation}
\sum_{i \in \X} \frac{1}{n}(\log K_i - \log L_i) \geq S(\X|\cX A)_{\psi} - O(1)\eta(2\sqrt{\epsilon(1-\epsilon/2)})
\end{equation}
for any non empty subset $\X \subseteq \cK$. Using a similar argumentation, we can show that
\begin{equation}
\sum_{i \in \Y} \frac{1}{n}(\log M_i - \log N_i) \geq S(\Y|\cY B)_{\psi} - O(1)\eta(2\sqrt{\epsilon(1-\epsilon/2)})
\end{equation}
holds for any non empty subset $\Y \subseteq \cKbar$. By letting $n \rightarrow \infty$ and $\epsilon \rightarrow 0$, we get the converse.
\end{proof}
If only a single copy of $\inputGroupstate$ is available to the involved parties, we can adapt the argument of Theorem \ref{thm:cost1} and prove the following result concerning the existence of split-transfer protocols with error $\epsilon$:
\begin{Proposition}
\label{prop:1ShotSplit}
Given a partition $\cK \subseteq \{1,2,\ldots,m\}$ of the helpers $C_1, C_2, \ldots, C_m$, let $\inputGroupstate$ be a $(m+3)$-partite pure state and fix
$\epsilon_1, \epsilon_2> 0$. Then, for any entanglement cost
$\overrightarrow{E_{\cK}} = \bigoplus_{i \in \cK}(\log K_i - \log L_i)$ and $\overrightarrow{E_{\cKbar}} = \bigoplus_{i \in \cKbar} (\log M_i - \log N_i)$ satisfying
 \begin{equation}\label{eq:cost}
   \begin{split}
   \log K_{{\cal S}} - \log L_{{\cal S}} := \sum_{i \in {\cal S}} (\log K_i - \log L_i) &\geq  -H_{\min}(\psi^{{\cal S} \cKbar B R}|\psi^{\cKbar B R}) + 4\log \left (\frac{1}{\epsilon_1} \right) + 2|\cK| + 8  \\
    \log M_{{\cal S}'} - \log L_{{\cal S}'} := \sum_{i \in {\cal S}'} (\log M_i - \log N_i) &\geq  -H_{\min}(\psi^{{\cal S}' A_{\cK} A R}|\psi^{A_{\cK} A R}) + 4\log \left(\frac{1}{\epsilon_2} \right) + 2|\cKbar| + 8 \\
     \end{split}
 \end{equation}
for all non-empty subsets ${\cal S} \subseteq \cK$ and ${\cal S}' \subseteq \cKbar$, there
exists a split-transfer protocol acting on $\inputGroupstate$ with error
$\epsilon_1+\epsilon_2$.
\end{Proposition}

\begin{proof}
The proof is very similar to the proof of Theorem \ref{thm:cost1}. First, we fix random measurements for each helper $C_i$ in a manner analogous to Proposition \ref{prop:isometry}. For each helper $C_i$ in $\cK$, we have $F_i = \lfloor \frac{d_{C_i}K_i}{L_i}\rfloor$ random partial isometries $Q^i_j U_i$ of rank $L_i$, where $Q^i_j$ is defined as in Proposition \ref{prop:isometry} and $U_i$ is a random Haar unitary on $C_iC^0_i$. If $F_iL_i < d_{C_i}K_i$, we also have a partial isometry of rank $L'_i < L_i$. Similarly, for each helper $C_i$ in $\cKbar$, we have $G_i = \lfloor \frac{d_{C_i}M_i}{N_i}\rfloor$ random partial isometries $Q^i_j U_i$ of rank $N_i$, and one of rank $N'_i$ if $G_iN_i < d_{C_i} M_i$. For a measurement outcome $J:=(j_1,j_2,\ldots, j_m)$, let $J_{\cK} = \bigoplus_{i \in \cK} j_i$ be the vector of length $t=|\cK|$ whose components correspond to the measurement outcomes for the helpers belonging to the cut $\cK$. The $i$-th element of $J_{\cK}$ will be denoted by $j_{\cK,i}$. Define
 \begin{equation}
   \omega_{J_{\cK}}^{\cK^1 \cKbar B R} := (Q^{\cK}_{J}U_{\cK} \otimes I_{\cKbar BR}) \psi^{\cK \cKbar BR} (Q_J^{\cK}U_{\cK} \otimes I_{\cKbar BR})^{\dag},
 \end{equation}
where $Q_{J}^{\cK} := \bigotimes_{i \in \cK} Q^i_{j_i}$ and the shorthand $U_{\cK}$ denotes the tensor product $\bigotimes_{i \in \cK}U_i$. If we apply Lemma \ref{Lemma:oneshotdecouple} to the state $\psi^{\cK R'} \otimes \tau^K$, where $\tau^K := \bigotimes_{i \in \cK} \tau^{K_i}$ and $R' := \cKbar \otimes B \otimes R$, we get
 \begin{equation}\label{eq:deriv2}
   \begin{split}
  \mathbb{E} \bigg [\sum_{j_{\cK,1}=1}^{F_1} \sum_{j_{\cK,2}=1}^{F_2} \cdots \sum_{j_{\cK,t}}^{F_t} \bigg \| &\omega^{\cK^1 R'}_{J}
- \frac{L}{d_{C_{\cK}}} \tau_L^{\cK^1} \otimes \psi^{R'} \bigg \|_1 \bigg ] \\ &\leq \frac{\prod_{i \in \cK} F_i
L_i}{d_{C_{\cK}}K} \sqrt{\sum_{\substack{{\cal S} \subseteq
\cK \\ {\cal S} \neq \emptyset}}
2^{-(H_{\mathrm{min}}(\psi^{{\cal S} R'}|\psi^{R'}) + \log K_{{\cal S}} -
   \log L_{{\cal S}})}} \\
  &\leq \sqrt{\sum_{\substack{{\cal S} \subseteq \cK \\ {\cal S}
   \neq \emptyset}}
   2^{-(H_{\mathrm{min}}(\psi^{{\cal S} R'}|\psi^{R'}) + \log K_{{\cal S}} -
   \log L_{{\cal S}})}}, \\
\end{split}
\end{equation}
where $K_{{\cal S}} = \prod_{i \in {\cal S}} K_i$.

Using the hypothesis that $\log K_{{\cal S}} - \log L_{{\cal S}} \geq -H_{\min}(\psi^{{\cal S} R'}|\psi^{R'}) + 4\log \left(\frac{1}{\epsilon_1} \right ) + 2|\cK| + 8$, we can proceed in a manner analogous to the proof of Theorem \ref{thm:cost1} and get the following bound on the expectation of the quantum error $Q^1_{\cal I}(\inputGroupstate \otimes \Phi^K)$:
\begin{equation}
\begin{split}
 \mathbb{E} \bigg [\sum_{j_{\cK,1}=0}^{F_1} &\sum_{j_{\cK,2}=0}^{F_2} \cdots \sum_{j_{\cK,t}=0}^{F_t} p_{J_{\cK}} \bigg \| \psi^{\cK^1 \cKbar B R}_{J_{\cK}}
- \tau^{\cK^1}_L \otimes \psi^{\cKbar B R} \bigg \|_1 \bigg ] \\
&\leq 2 \sum_{\substack{{\cal S} \subseteq \cK \\ {\cal
S} \neq \emptyset}}\prod_{i \in {\cal S}} \frac{L_i}{d_{C_i}K_i} +
\mathbb{E} \bigg [{\sum_{j_{\cK,1}=1}^{F_1} \sum_{j_{\cK,2}=1}^{F_2} \cdots
\sum_{j_{\cK,t}=1}^{F_t} p_{J_{\cK}} \bigg \| \psi^{\cK^1 \cKbar B R}_{J_{\cK}}
- \tau^{\cK^1}_L \otimes \psi^{\cKbar B R} \bigg \|_1 } \bigg ]\\
 &\leq  \sum_{\substack{{\cal S} \subseteq \cK \\ {\cal
S} \neq \emptyset}}\frac{2\epsilon_1^4 2^{H_{\min}(\psi^{\cal S})}}{2^{2t+8}d_{C_{{\cal S}}}} + \frac{\epsilon^2_1}{8}\\
& \leq  \frac{\epsilon_1^4}{2^{t+7}} + \frac{\epsilon^2_1}{8} \leq \frac{\epsilon_1^2}{4}, \\
 \end{split}
 \end{equation}
 where $t=|\cK|, p_{J_{\cK}} = \Tr(\omega^{\cK^1 \cKbar B R}_{J_{\cK}})$ and $\psi^{\cK^1 \cKbar B R}_{J_{\cK}} = \frac{1}{p_{J_{\cK}}}\omega^{\cK^1 \cKbar B R}_{J_{\cK}}$. In a similar way, we can bound the expected value of the quantum error $Q^2_{\cal I}$ as follows:
\begin{equation}
\begin{split}
 \mathbb{E} \bigg [\sum_{j_{\cKbar,1}=0}^{G_1} \sum_{j_{\cKbar,2}=0}^{G_2} \cdots \sum_{j_{\cKbar,m-t}=0}^{G_{m-t}} p_{J_{\cKbar}} \bigg \| \psi^{A_{\cK} \cKbar^1 A R}_{J_{\cKbar}}
- \tau^{\cKbar^1}_N \otimes \psi^{A_{\cK} A R} \bigg \|_1 \bigg ]
 &\leq  \sum_{\substack{{\cal S}' \subseteq \cKbar \\ {\cal
S}' \neq \emptyset}}\frac{2\epsilon_2^4 2^{H_{\min}(\psi^{{\cal S}'})}}{2^{2(m-t)+8}d_{C_{{\cal S}'}}} + \frac{\epsilon^2_2}{8}\\
& \leq  \frac{\epsilon^4_2}{2^{m-t+7}} + \frac{\epsilon^2_2}{8} \leq \frac{\epsilon^2_2}{4}. \\
 \end{split}
 \end{equation}
From Proposition \ref{prop:mergeConds}, we can conclude that there exists a split-transfer protocol of error $\epsilon_1+\epsilon_2$. \end{proof}
With these results in hand, we can now return to our initial motivation, which was that of proving that the min-cut entanglement of the state $\inputstate$ can be preserved by letting all the helpers $C_1,C_2,\ldots,C_m$ perform simultaneous random measurements on their typical subspaces. To prove this fact, we will need the following corollary to Theorem \ref{thm:splittransfer}.

\begin{Corollary}\label{cor:LOCC}
For a pure state $\inputstate$, we denote by $\cK_{\min}$ a cut of
the smallest possible size with the following property: \be
\forall \cK \subseteq \{C_1,C_2,\ldots,C_m\}:
S(A\cK_{\min})_{\psi} \leq S(A\cK)_{\psi}.\ee Then, for the state
$\psi^{\cK_{\min}\cKbar_{\min} AB}$, the right hand side of
eq.~(\ref{eq:splitcond1}) will be negative for all nonempty sets ${\cal X}
\subseteq \cK_{\min}$, while the right hand side of
eq.~(\ref{eq:splitcond2}) will be non-positive for all nonempty sets ${\cal Y}
\subseteq \cKbar_{\min}$.

Furthermore, if we have arbitrarily many copies of the state $\inputstate$ at our disposal, we can perform a split-transfer of the state $\inputstate$ using only local operations and classical communication.
\end{Corollary}

\begin{proof}
For any non-empty subset ${\X} \subseteq {\cal T}_{min}$, where $\cK_{\min}$ is not the empty set, we have
 \begin{equation}
\begin{split}
S(\X|\cX A)_{\psi} &= S(\cK_{\min}A)_{\psi} - S(\X \cKbar_{\min}B)_{\psi} \\
& < S(\cX A)_{\psi} - S(\X \cKbar_{\min} B)_{\psi}\\
&= S(\X \cKbar_{\min} B)_{\psi} - S(\X \cKbar_{\min} B)_{\psi} \\
&= 0,\\
\end{split}
\end{equation}
where in the second line we have used the fact that $S(A \cK_{\min})_{\psi} < S(A \cK)$ when $\cK$ is a cut of size smaller than $|\cK_{\min}|$.

Similarly, for any non-empty subset $\Y \subseteq \cKbar_{\min}$, where $\cK_{\min}$ is not the whole set $\{C_1,C_2,\ldots,C_m\}$, we have
\begin{equation}
  \begin{split}
S(\Y|\cY B)_{\psi} &= S(\cKbar_{\min}B)_{\psi} - S(\Y \cK_{\min} A)_{\psi} \\
&= S(\cK_{\min}A)_{\psi} - S(\Y \cK_{\min} A)_{\psi} \\
&\leq 0 \\
  \end{split}
\end{equation}
This proves the first part of the corollary.

To get the second part, apply Theorem \ref{thm:splittransfer} by
setting the Schmidt ranks of the pre-shared maximally entangled
states to be $K_i = 1$ for all $i \in \cK$ and $N_j=1$ for all $j
\in \cKbar$. Then, for these particular values,
eqs.~(\ref{eq:bound1}) and (\ref{eq:bound2}) give us bounds on the
ranks $L_i$ and $M_j$ of projectors corresponding to measurements
performed by $C_i \in \cK$ and $C_j \in \cKbar$ respectively.
Since $L_i \geq 1$ and $M_j \geq 1$ must be satisfied for all $i
\in \cK$ and $j \in \cKbar$, we need the conditional entropies
$S(\X|\cX A)_{\psi}$ and $S(\Y|\cY B)_{\psi}$ appearing in the
upper bounds to $\prod_{i \in \cK} L_i$ and $\prod_{j \in \cKbar}
M_j$ to be negative. Otherwise, the helpers will not be able to
perform measurements with vanishing quantum errors $Q^1_{\cal I}$
and $Q^2_{\cal I}$ and they will need to consume additional
entanglement.

If some of the conditional entropies $S(\Y|\cY B)_{\psi}$ are equal to zero, we will need to inject an arbitrarily small amount of singlets between the cut $\Y$ vs $A \cK_{\min}$ or the cut $\Y$ vs $B \cY$ in order to make $S(\Y|\cY B)_{\psi}$ negative (i.e an EPR pair contributes -1 to the conditional entropy). However, it is shown in \cite{Smolin} that for pure states, the LOCC class of transformations is not more powerful if we allow an additional sublinear amount of entanglement. This is due to the fact that we can always generate EPR pairs between a given cut, using an $o(n)$ amount of copies of the initial state, unless across that cut the state happens to be in a product state.
\end{proof}

\begin{Theorem}[Multipartite Entanglement of Assistance \cite{Merge}]
Let $\inputstate$ be a state shared between $m$ helpers and two recipients: Alice and Bob. Given many copies of $\inputstate$, if we allow LOCC operations between the helpers and the recipients, the optimal "assisted" EPR rate is given by  \begin{equation}\label{eq:mincut}
  E^{\infty}_A(\psi,A:B) = \min_{\cK} \{S(A\cK)_{\psi}  \}
\end{equation}
\end{Theorem}
\begin{proof}
Let $\cK_{\min}$ be a cut of the smallest size attaining the
minimization in eq.~(\ref{eq:mincut}) and fix some $\epsilon > 0$.
Then, according to Corollary \ref{cor:LOCC}, if $n$ is large
enough, we can perform a split-transfer protocol of the state
$\psi^{\cK_{\min}\cKbar_{\min}AB}$ with error $\epsilon$. This
will produce a state
$\varphi^{A^nA^n_{\cK_{\min}}B^nB^n_{\cKbar_{\min}}}$ such that
\begin{equation}\label{eq:dist}
 \bigg \| \varphi^{A^n A^n_{\cK_{\min}}} - (\psi^{AA_{\cK_{\min}}})^{\otimes n} \bigg \|_1 \leq \bigg \| \varphi^{A^nA^n_{\cK_{\min}}B^nB^n_{\cKbar_{\min}}} - (\psi^{AA_{\cK_{\min}}BB_{\cKbar_{\min}}})^{\otimes n} \bigg \|_1 \leq \epsilon,
\end{equation}
where $\psi^{AA_{\cK_{\min}}BB_{\cKbar_{\min}}}$ is the original
state $\psi^{\cK_{\min}\cKbar_{\min}AB}$ with the systems
$A_{\cK_{\min}}$ and $B_{\cKbar_{\min}}$ substituted for the
systems $\cK_{\min}$ and $\cKbar_{\min}$. Applying the Fannes
inequality to eq.~(\ref{eq:dist}), we get
 \begin{equation}
   \bigg | S(A^nA^n_{\cK_{\min}})_{\varphi} - n(S(AA_{\cK_{\min}})_{\psi} \bigg | \leq n \log (d_A d_{A_{\cK_{\min}}}) \eta(\epsilon)
 \end{equation}
which implies that
 \begin{equation}
   S(A^nA^n_{\cK_{\min}})_{\varphi} = n(S(AA_{\cK_{\min}})_{\psi} \pm \delta) = n(S(A\cK_{\min}) \pm \delta),
 \end{equation}
 where $\delta$ can be made arbitrarily small by letting $\epsilon \rightarrow 0$.  Thus, the min-cut entanglement $E^{\infty}_A(\psi)$ is arbitrarily well preserved after the split-transfer is performed, and so Alice and Bob can distill at this rate by applying a standard purification protocol on $\varphi^{A^nA^n_{\cK_{\min}}B^nB^n_{\cKbar_{\min}}}$.
\end{proof}

\section{Discussion} \label{sec:discussion}

We have studied the problem of multiparty state merging with an emphasis on
how to accomplish merging when the participants have access only to a single copy of a quantum state.
In the easier asymptotic i.i.d. setting, the rate region was
characterized by a set of  ``entropic'' inequalities which any
rate-tuple $(R_1,R_2,\ldots,R_m)$ must satisfy in order to be
achievable for merging. These inequalities define a convex region
$S$ in an $m$-dimensional space, whose axes are the individual
rates $R_i$, and where merging can be achieved if the parties have
access to many copies of $\psi^{C_1C_2\ldots C_m BR}$. Our
protocol for multiparty state merging distinguishes itself in that
any point in the rate region can be achieved without the need for time-sharing. The
main technical challenge for showing this was to adapt the
decoupling lemma of \cite{Merge} and the upper bound to the
quantum merging error (Proposition 4 in \cite{Merge}) to the
multiparty setting.

The one-shot analysis of the entanglement cost necessary to perform
merging presented more difficulties
than in the asymptotic setting but as compensation yielded greater rewards.
Most notably, because time-sharing is impossible with only a single copy of a quantum state,
our intrinsically multiparty protocol provides the first method to interpolate between achievable
costs in the multiparty setting. The technical challenge was to derive an upper bound on
the quantum error $Q_{\cal I}(\psi)$ for a random coding
strategy in terms of the min-entropies. We suspect
that it might possible to further improve our bound by replacing the
min-entropies with their smooth variations, but it is unclear how to proceed in order to show this. We leave it as an open
problem. To illustrate the advantages of intrinsic multiparty merging over iterated two-party merging, we also performed a detailed analysis of the costs incurred by the two strategies for variants of the embezzling states.

Lastly, we have introduced the split-transfer problem, a variation
on the state merging task, and applied it in the context of
multiparty assisted distillation. The main technical difficulty
here was to prove that the helpers in the cut $\cKbar$ do not have
to wait for the helpers in $\cK$ to complete their merging with
the decoder $A$ before they can proceed with the transfer of their
shares to the $B$ decoder. The essential ingredients for showing
this were the commutativity of the Kraus operators $P^{\cK}_{\jK}$
and $P^{\cKbar}_{\jKbar}$, and the triangle inequality.  The rate
region for a split-transfer is composed of two sub-regions, each
corresponding to rates which would be achievable for a merging
operation from $\cK$(resp. $\cKbar$) to $A$(resp. $B$) with
reference $\cKbar BR$ (resp. $\cK AR$).

In the context of assisted distillation, the existence of a
split-transfer protocol which redistributes the initial pure state
$\psi^{C_1C_2 \ldots C_m AB}$ to the decoders $A$ and $B$ was used
to give a non-recursive proof that the optimal achievable EPR rate
under assistance is given by the min-cut entanglement
$\min_{\cK}\{S(A\cK)\}$. It would be interesting to come up with
other potential applications for the split-transfer protocol.
State merging was used as a building block for solving various
communication tasks, and we believe split-transfer
could be useful in other multipartite scenarios than the assisted
distillation context. Alternatively, it could also simplify some
of the existing protocols which rely on multiple applications of
the state merging primitive.

\acknowledgments
The authors would like to thank J\"urg Wullschleger for an interesting discussion on the subject of time-sharing and Andreas Winter for discussions on multiparty state transfer.
This research was supported by the Canada Research Chairs
program, CIFAR, FQRNT, INTRIQ, MITACS, NSERC, ONR grant
No.~N000140811249 and QuantumWorks.

\appendix

\section{Miscellaneous Facts} \label{app:facts}
    For an operator $X$, the trace norm is defined as: \[ \|X\|_1 := \Tr \sqrt{X^{\dag}X}, \]
    and the trace distance of two states $\rho$ and $\sigma$ is given by $D(\rho,\sigma) = \frac{1}{2}\|\rho-\sigma\|_1$.
     An alternative measure of closeness of two states is given by the fidelity:
     \[ F(\rho,\sigma) := \biggl ( \Tr\sqrt{\sqrt{\rho}\sigma\sqrt{\rho}}\biggr ). \]
    If the state $\rho := \braket{\psi}$ is pure, the fidelity between $\rho$ and $\sigma$ becomes equal to:
    \[ F(\braket{\psi},\sigma) = \sqrt{ \langle \psi | \sigma | \psi \rangle} = \sqrt{ \Tr(\rho \braket{\psi})}  \]
   These two measures of closeness are related as follows:
   \begin{Lemma}\cite{FuchsVandegraaf:fidelity}\label{Lemma:relation}
   For states $\rho$ and $\sigma$, the trace distance is bounded by
         \[ 1- F(\rho, \sigma) \leq D(\rho,\sigma) \leq \sqrt{1-F^2(\rho,\sigma)}. \]
   \end{Lemma}

   \begin{Lemma}[Fannes Inequality \cite{Fannes}]\label{lem:Fannes}
     Let $\rho$ and $\sigma$ be states on a $d$-dimensional Hilbert space, with $\|\rho - \sigma\|_1 \leq \epsilon$. Then $|H(\rho) - H(\sigma)| \leq \eta(\epsilon) \log{d}$, where $\eta(x) = x - x \log x$ for $x \leq \frac{1}{\epsilon}$. When $x > \frac{1}{\epsilon}$, we set $\eta(x)=x + \frac{\log \epsilon}{\epsilon}$.
   \end{Lemma}
\begin{Lemma}[Gentle Measurement Lemma \cite{Winter02}]
Let $\rho$ be a subnormalized state (i.e $\rho \geq 0$ and $\Tr[\rho] \leq 1$). For any operator $0 \leq X \leq I$ such that $\Tr[X \rho] \geq 1 -\epsilon$, we have
\[
  \bigg \| \sqrt{X} \rho \sqrt{X} - \rho \bigg \|_1 \leq 2 \sqrt{\epsilon}
\]
\end{Lemma}

\section{Proof of eq.~(\ref{eq:operatorineq})} \label{app:operatorineq}
\begin{comment}
Let $\rho^{F} = \sum_x p(x) \braket{x}$ be a density operator on the space $F$, where $\{\ket{x}\}$ is an orthonormal set of vectors and $\{p(x)\}$ are the corresponding eigenvalues. Given $n$ identical copies of $\rho$, written as $\rho^{\otimes n} = \sum_{x_1x_2\ldots x_n} p(x_1)p(x_2)\ldots p(x_n) \braket{x_1x_2\ldots x_n}$, we say the state $\ket{x_1x_2\ldots x_n}$ is $\delta$-typical if we have
\[
  \bigg | \frac{1}{n} \log \bigg ( \frac{1}{p(x_1)p(x_2)\ldots p(x_n)} \bigg ) - S(\rho) \bigg | \leq \delta,
\]
for any $\delta > 0$. The typical subspace $\ti{F}$ associated with the density operator $\rho$ is the subspace of $F$ spanned by all $\delta$-typical states $\ket{x_1x_2\ldots x_n}$, with corresponding projector \[ \Pi_{\ti{F}} = \sum_{\substack{x_1x_2\ldots x_n \\ \mathrm{\phantom{=}\delta-typical}}} \braket{x_1x_2\ldots x_n}. \]
For any $\epsilon > 0$ and large enough $n$, we have
 \begin{equation}
        \begin{split}
          \Tr (\Pi_{\ti{F}} \rho^{\otimes n}) &\geq 1- \epsilon,\\
         (1-\epsilon)2^{n(S(F)_{\rho} - \delta)} &\leq \Tr( \Pi_{\ti{F}}) \leq 2^{n(S(F)_{\rho}+\delta)}, \\
          \Tr [ (\rho^F)^2 ] &\leq 2^{-n(S(F)_{\rho} - \delta)}. \\
         \end{split}
  \end{equation}
For the first two inequalities, see \cite{Nielsen} for a proof. The last inequality is shown in \cite{Hayden001}.
\end{comment}

\begin{Lemma}
For $n$ copies of a state $\initstate$, let $\Pi_{\ti{B}},\Pi_{\ti{C}_1},\Pi_{\ti{C}_2},\ldots,\Pi_{\ti{C}_m},\Pi_{\ti{R}}$ be the projectors onto the typical subspaces $\ti{B},\ti{C}_1,\ti{C}_2,\ldots,\ti{C}_m$ and $\ti{R}$ respectively. Then, we have
\begin{equation}\label{eq:operatorineq1}
    \Pi_{\ti{B}\ti{C}_M\ti{R}} := \Pi_{\ti{B}} \otimes \Pi_{\ti{C_1}} \otimes \ldots \otimes \Pi_{\ti{C_m}} \otimes \Pi_{\ti{R}} \geq \Pi_{\ti{B}} +
   \Pi_{\ti{C_1}} + \ldots + \Pi_{\ti{C_m}} + \Pi_{\ti{R}} - (m+1) I_{\ti{B}\ti{C}_M\ti{R}}, \\
 \end{equation}
 where $\Pi_{\ti{B}}$ is a shorthand for $\Pi_{\ti{B}} \otimes I^{C_1C_2\ldots C_mR}$, and similarly for $\Pi_{\ti{C}_1},\Pi_{\ti{C}_2},\ldots,\Pi_{\ti{C}_m}$ and $\Pi_{\ti{R}}$.
\end{Lemma}
\begin{proof}
 The projection operators involved in the proof statement pairwise commute, and thus, are simultaneously diagonalizable. Let $\{\ket{e_i}\}$ be a common eigenbasis for these projectors. Then any eigenvector $\ket{e_i}$ with $\Pi_{\ti{B}\ti{C}_M\ti{R}}\ket{e_i} = \ket{e_i}$ satisfies
 \[
 \bigg (\Pi_{\ti{B}} + \Pi_{\ti{C_1}} + \ldots + \Pi_{\ti{C_m}} + \Pi_{\ti{R}} - (m+1) I \bigg ) \ket{e_i} = \ket{e_i}.
 \]
 If $\ket{e_i}$ is any eigenvector with $\Pi_{\ti{B}\ti{C}_M\ti{R}}\ket{e_i} = 0$, then it must be in the kernel of at least one of the projection operators $\Pi_{\ti{B}},\Pi_{\ti{C}_1},\Pi_{\ti{C}_2},\ldots,\Pi_{\ti{C}_m}$ and $\Pi_{\ti{R}}$, which implies that
\[
\bigg (\Pi_{\ti{B}} +
   \Pi_{\ti{C_1}} + \ldots + \Pi_{\ti{C_m}} + \Pi_{\ti{R}} - (m+1) I_{\ti{B}\ti{C}_M\ti{R}} \bigg )\ket{e_i} = \lambda_i \ket{e_i},
\]
where $\lambda_i \leq 0$.
Using both of these observations, we have
\begin{equation}
 \begin{split}
   \Pi_{\ti{B}\ti{C}_M\ti{R}} = \sum_{ \Pi_{\ti{B}\ti{C}_M\ti{R}}\ket{e_i}=\ket{e_i}} \braket{e_i} & \geq \sum_{ \Pi_{\ti{B}\ti{C}_M\ti{R}}\ket{e_i}=\ket{e_i}} \braket{e_i} + \sum_{ \Pi_{\ti{B}\ti{C}_M\ti{R}}\ket{e_i}=0} \lambda_i \braket{e_i} \\
   &= \Pi_{\ti{B}} +
   \Pi_{\ti{C_1}} + \ldots + \Pi_{\ti{C_m}} + \Pi_{\ti{R}} - (m+1) I_{\ti{B}\ti{C}_M\ti{R}} \\
 \end{split}
\end{equation}
\end{proof}

\section{Smoothing $H_\max$} \label{app:hmax-smoothing}
\begin{Lemma} \label{lem:hmax-smoothing}
Suppose the density operator $\rho$ has eigenvalues $r = (r_1, \ldots, r_d)$ with $r_j \geq r_{j+1}$. Then
\begin{equation}
H^\epsilon_\max(\rho) \geq 2 \log \min \left\{ \sum_{j=1}^{k-1} \sqrt{r_j}:
	k \mbox{ such that } \sum_{j=k+1}^d r_j \leq \frac{\epsilon^2}{2} \right\}.
\end{equation}
\end{Lemma}
\begin{proof}
By Lemma 16 of \cite{Renner01}, $H^\epsilon_\max(\rho)$ is equal to the minimum of $H_\max(\overline\rho)$ over all positive semidefinite operators $\overline\rho$ no more than $\epsilon$ away from $\rho$ as measured by the purified distance. This measure is a bit awkward to work with for our purposes, but it is bounded above by $\sqrt{2} \| \rho - \overline\rho \|_1^{1/2}$ by eq.~(\ref{eq:purified}).
%(Note that their quantity $\overline{D}(\rho,\overline\rho)$ is bounded above by $\|\rho-\overline\rho\|_1$.)
Therefore,
\begin{eqnarray}
H^\epsilon_\max(\rho)
	&\geq& \min \left\{
	H_\max(\overline\rho) : \| \rho - \overline\rho \|_1 \leq \delta
	\right\}
\end{eqnarray}
for $\delta = \epsilon^2/2$
and we will try to estimate the right hand side of the inequality.

 Let $\overline\rho$ be a positive semidefinite operator such that $\| \overline\rho -\rho \|_1 \leq \delta$ and let $\overline{r} = ( \overline{r}_1,\ldots,\overline{r}_d )$ be the eigenvalues of $\overline\rho$, ordered such that $\overline{r}_j \geq \overline{r}_{j+1}$. We will identify $r$ and $\overline{r}$ with their associated diagonal matrices. Then (see \cite{Nielsen})
\begin{equation}
\| \overline{r} - r \|_1 \leq \| \overline\rho - \rho \|_1,
\end{equation}
but $H_\max(r) = H_\max(\rho)$ and $H_\max(\overline{r}) = H_\max(\overline\rho)$
so we may assume without loss of generality that $\overline\rho$ and $\rho$ are simultaneously diagonal with diagonal entries in non-increasing order. We can therefore dispense with $\rho$ and $\overline\rho$, discussing only $r$ and $\overline{r}$ from now on.

By Theorem 3 of \cite{Renner03}, $H_\max(r) = 2 \log \sum_j \sqrt{r_j}$, which is monotonically decreasing in each $r_j$. This implies that a minimizing $\overline{r}$ must satisfy $r_j \geq \overline{r}_j$. If not, redefining $\overline{r}_j = r_j$ decreases $\| \overline\rho - \rho \|_1$ and $H_\max(\overline{r})$ at the same time.

We will now argue that there is a minimizing $\overline{r}$ such that there is a $j_0$ for which $r_j = \overline{r}_j$ for all $j < j_0$ and $\overline{r}_j = 0$ for all $j > j_0$. Let $s = (s_1,\ldots,s_d)$ be any vector such that $s_j \geq s_{j+1} \geq 0$ and $s_j \leq r_j$, that is, a vector that is a possible candidate for a minimizer. Suppose that $s$ does not have the prescribed form, that is, there is a $j_0$ such that $s_{j_0} < r_{j_0}$ but $s_{j_0+1} \neq 0$. Consider the family of vectors $t(\gamma)$ that arise by transferring $\gamma$ from $s_{j_0+1}$ to $s_{j_0}$ defined by $t(\gamma)_{j_0} = s_{j_0} + \delta$, $t(\gamma)_{j_0+1} = s_{j_0+1} - \delta$ and $t_j = s_j$ for $j \not\in \{ j_0, j_0+1 \}$.

It is easy to check that for sufficiently small $\gamma$, it will be the case that $\| r - t(\gamma) \|_1 \leq \| r - s \|_1$. Moreover, defining $f(\gamma) = \sum_j \sqrt{t(\gamma)_j}$, we have that
\begin{equation}
\frac{df}{d\gamma}\Big|_{\gamma = 0}
	= \frac{1}{2\sqrt{s_{j_0}}} - \frac{1}{2 \sqrt{s_{j_0+1}}},
\end{equation}
which is nonpositive since $s_{j_0} \geq s_{j_0+1}$. For sufficiently small $\gamma$ then, $H_\max(t(\gamma)) \leq H_\max(s)$. (If $s_{j_0} = 0$, the derivative does not exist but the conclusion can be confirmed by looking at finite differences.) So, if $s$ were a minimizer, it is possible to  either construct a new minimizer of the prescribed form or reach a contradiction by further decreasing $H_\max$.

The statement of the lemma follows by evaluating $H_\max$ on a minimizer of the prescribed form.
\end{proof}

\bibliographystyle{unsrt}
\bibliography{EofABib}
\end{document}